\documentclass[english,12pt]{article}
\usepackage{eurosym}
\usepackage{times}
\usepackage{amsfonts}
\usepackage{amsmath,amssymb}
\usepackage{amsthm}
\usepackage{natbib}
\usepackage{bm}
\usepackage{graphicx}
\usepackage{subcaption}
\usepackage{appendix}
\usepackage[figuresright]{rotating}
\usepackage{ctable}
\usepackage{array}
\usepackage{tabularx}
\usepackage{tabulary}
\usepackage{footnote}
\usepackage{threeparttable}
\usepackage{float}
\usepackage{mathrsfs}
\usepackage{dsfont}
\usepackage{nicefrac}
\usepackage{bm}
\usepackage{xspace}
\usepackage{color}
\usepackage[colorlinks,linkcolor=blue,citecolor=blue,bookmarks=false,pagebackref]{hyperref}
\usepackage[top=1.0in, bottom=1in, left=0.7in, right=0.9in]{geometry}
\usepackage{setspace}
\usepackage{titling}
\usepackage{mathabx}
\usepackage{placeins}

\allowdisplaybreaks
\theoremstyle{plain}
\newtheorem{theorem}{Theorem}[section]
\theoremstyle{plain}
\newtheorem{lemma}{Lemma}[section]
\theoremstyle{plain}
\newtheorem{assumption}{Assumption}[section]
\theoremstyle{plain}

\theoremstyle{plain}
\newtheorem{corollary}{Corollary}[section]
\theoremstyle{remark}
\newtheorem{remark}{Remark}[section]
\theoremstyle{comment}

\theoremstyle{definition}

\theoremstyle{definition}

\newcommand{\plim}{\operatornamewithlimits{plim\,}}

\mathchardef\mhyphen="2D

\floatstyle{ruled}
\newfloat{algorithm}{tbp}{loa}
\providecommand{\algorithmname}{Algorithm}
\floatname{algorithm}{\protect\algorithmname}

\begin{document}

\title{Solving the Forecast Combination Puzzle}
\author{David T.~Frazier\thanks{
Department of Econometrics and Business Statistics, Monash University,
Melbourne, Australia: \texttt{david.frazier@monash.edu} }, Ryan Covey,
Gael M. Martin and Donald Poskitt}
\maketitle

\begin{abstract}
We demonstrate that the forecasting combination puzzle is a consequence of
the methodology commonly used to produce forecast combinations. By the
combination puzzle, we refer to the empirical finding that predictions
formed by combining multiple forecasts in ways that seek to optimize
forecast performance often do not out-perform more naive, e.g.
equally-weighted, approaches. In particular, we demonstrate that, due to the
manner in which such forecasts are typically produced, tests that aim to
discriminate between the predictive accuracy of competing combination strategies
can have low power, and can lack size control, leading to an outcome that
favours the naive approach. We show that this poor performance is due
to the behavior of the corresponding test statistic, which has a
non-standard asymptotic distribution under the null hypothesis of no inferior predictive accuracy, rather than
the {standard normal distribution that is} {typically adopted}. In
addition, we demonstrate that the low power of such predictive accuracy
tests in the forecast combination setting can be completely avoided if more
efficient estimation strategies are used in the production of the
combinations, when feasible. We illustrate these findings both in the
context of forecasting a functional of interest and in terms of predictive
densities. A short empirical example {using daily financial returns} exemplifies how researchers can avoid the puzzle in practical settings.
\end{abstract}

%\date{}

\section{Introduction}

Since their inception (\citealp{Stone1961} and \citealp{Bates1969}),
forecast combination methods have garnered a dedicated following due to
their flexibility and accuracy (\citealp{Timmermann2006}; %
\citealp{Aastveit2019}). Such methods also align with the zeitgeist of
modern econometric thought, in that they are designed to accommodate the
fact that not all data sources are created equal, and that the models we are
working with in economics are at best an approximation to reality.

Generally, forecast combinations are constructed by producing forecasts for
individual, or constituent models, and then combining them via some
combination function or weighting scheme. \textit{Point} forecast
combinations are typically constructed by taking a weighted average of point
forecasts produced by the constituent models (\citealp{Bates1969}; %
\citealp{Stock2004}; \citealp{Timmermann2006}; \citealp{Smith2009}; %
\citealp{Claeskens2016}). In the case of \textit{distributional} forecast
combinations, two commonly used approaches are the linear opinion pool (%
\citealp{Stone1961}; \citealp{Hall2007}; \citealp{Geweke2011}; %
\citealp{Opschoor2017}; \citealp{Martin2021}), and the beta-transformed
linear opinion pool (\citealp{Ranjan2010}; \citealp{Gneiting2013}; %
\citealp{Satopaeae2014}; \citealp{Baran2018}). The weighted average, linear
pool and beta-transformed linear pool are all \textit{combination functions}
that map a set of forecasts, produced using constituent models, to a single
forecast combination.

Even though the last fifty years has seen these methods rise to prominence
among empirical forecasters (\citealp{Makridakis2018}; \citealp{Thorey2018}; %
\citealp{Wang2018}; \citealp{Makridakis2020}; \citealp{Taylor2020}), key
issues regarding the use and abuse of the methods remain. One of the most
interesting issues is the so-called `forecast combination puzzle', which is
a stylized fact that states that predictions produced using complicated
combinations of different forecasts, e.g., via optimizing the weights in the
combination using some criterion function that encapsulates some aspect of
forecast accuracy, do not generally outperform simpler procedures (such as
equal-weighted combinations). {For example, see \cite{Stock2004}, \cite%
{Smith2009}, \cite{Makridakis2018, Makridakis2020} for empirical evidence of
this phenomenon.}

Explanations for the puzzle range from the increased sampling variability of
complex weighting schemes (\citealp{Stock2004}; \citealp{Claeskens2016}), to
the similar performance of equally-weighted and optimally-weighted
combinations (\citealp{Elliott2011}), and to bias in the averages loss
functions (\citealp{Chan2018}). However, all of the above explanations are
specific to linear combinations of point forecasts evaluated according to
mean squared forecast error. To date, there is no single universally
accepted answer that sufficiently explains this puzzle across both point and
distributional forecasts, or across different performance measures. See \cite%
{graefe2014combining} for a historical summary of the puzzle.

Herein, we analyze the forecast combination puzzle in general terms, and
through the lens of tests of superior forecast accuracy (\citealp{White2000}%
; \citealp{Hansen2005}). We demonstrate that the standard two-step approach
to producing forecast combinations results in tests that have \textit{no
power} against a large class of alternatives, including fixed-, random-, and
drifting-weighting schemes. This result is a consequence of the following
unintuitive feature of forecast combinations: when produced in the standard
manner, the {usual test statistic employed to gauge differences in forecasts
is such that estimated combination weights imparts \textit{no} sampling
variability into the statistic at first-order}. This finding extends and
renders rigorous previous results documented in a variety of contexts,
including by \cite{Stock2004}, and \cite{Smith2009}. In particular, \cite%
{Smith2009} argue that \textquotedblleft the forecast combination puzzle
rests on a gain ... that has no practical significance\textquotedblright ,
and our results make rigorous this statement by showing that, in general
settings, tests of superior forecast accuracy cannot distinguish between
estimated and a large class of forecast combinations, including the
equally-weighted combination.

More generally, we show that in order for the difference between two (sets
of) forecast combinations to be meaningful, the combination weights must be
surprisingly disparate. Critically, the distance between two combination
weights needed to yield a test with non-trivial power depends entirely on
the chosen loss and the variability in the constituent model forecasts.
Therefore, when comparing two forecast combinations obtained using \textit{%
the same} constituent models \textit{but different} combination weights, the
variability in the constituent forecasting models can easily swamp large
differences in the forecast distribution due to differences in the
combination weights. Consequently, a statistically significant difference
between two forecast combinations is unlikely to eventuate in practice
unless we have: 1) large sample sizes, much larger than typically used in
empirical applications; and 2) constituent model forecasts that have low
variability, which ultimately requires that the estimators of the 
underlying parameters of the constituent models also have low variability.
Consequently, in empirical applications that require the use of
high-variance constituent model forecasts, it is \textit{unlikely} that we
will be able to detect differences amongst different forecast combination
methods.

We demonstrate that the poor behavior of forecast accuracy tests in this
setting results from the test statistic having a non-standard asymptotic
distribution under the null of no inferior predictive accuracy. {That is, w}%
hile the critical values for such tests are {typically }based on the
standard normal distribution, we show that in the case of forecast
combination methods, an appropriately scaled version of the test statistic
converges in distribution to a generalized chi-squared distribution.
Furthermore, we show that this result persists across a large class of
forecast combination methods, including those with time-varying weights.

Lastly, we demonstrate that the forecast combination puzzle can be
circumvented in cases where it is feasible to produce forecast combinations
in a single step. That is, except in trivial cases, forecast combinations
produced in a single step do not exhibit the forecast combination puzzle,
and will always yield superior forecast accuracy over standard, or
equally-weighted, forecast combination schemes. In this way, we build on an
extensive investigation of one- and two-step forecast combinations in \cite%
{Zischke2022}, which also provides support for the one-step approach. We
reiterate that it is the two-step approach to producing forecast
combinations that is the standard approach adopted in the literature on
(frequentist) forecast combinations and, hence, the reason why the
combination puzzle has been so empirically prevalent. Revisiting the S\&P500 returns density prediction example considered in \cite{Geweke2011}, we show that: 1) the forecast combination puzzle is evident in their original example for certain classes of volatility models; {and} 2) the puzzle is entirely resolved if we use forecast combinations that are built in a single step.

Before moving on, we note here the following notations used throughout the
remainder of the paper. For a probability measure $P$, a random variable $X$
and a random sequence $X_n$, we write $\mathbb{E}_P[X]$ to denote the
expectation of $X$ under $P$, $\operatornamewithlimits{plim\,}_n X_n$ to
denote the probability limit of $X_n$ as $n \to \infty$ (if it exists), and $%
X_n \Rightarrow X$ if $X_n$ converges in distribution to $X$. For some
positive sequence $R_n$, the notations $X_n = o_p(R_n)$ and $X_n = O_p(R_n)$
have their usual definitions, see \cite{VanDerVaart1998} for a textbook
treatment. We say that $X_n \asymp R_n$ if there exists constants $c$ and $C$
such that $c R_n \leq X_n \leq C R_n$ for all $n$ large enough (with
probability one). The gradient and hessian of a functional $f$ of $x$ is
written $\nabla_{x} f(x)$ and $\nabla^2_{xx} f(x)$, respectively, where for $%
x$ on the boundary of the domain of $f$ these symbols denote the left or
right derivatives, whichever of those exist at $x$.

\section{A Brief Motivating Example}

\label{sec:motiv} We first motivate our analysis by reconsidering the
results and discussion in \cite{Smith2009} and demonstrating that the
forecast combination puzzle extends far beyond their initial analysis.

\subsection{Revisiting \protect\cite{Smith2009}}

\label{subsec:revisitingsmith}

Our goal is to produce a point forecast for the random variable $Y_{t}$ at
time $t=T+1$ using observed data $\{y_t:1\le t\le T\}$ and models $f_{j}$, $%
j=1,2$. Denote the point predictions from these models by $\tilde{y}_{jt}$,
and define the linear combination of point predictions by $\tilde{y}%
_{t}^{\eta}=\eta \tilde{y}_{1t}+(1-\eta)\tilde{y}_{2t}$ where $0\le\eta\le 1$%
, and assume that the parameters underlying the models $f_1$ and $f_2$ have
been estimated in a first stage. The accuracy of the forecasts is measured
according to mean squared forecast error (MSFE).

Under the assumptions in \cite{Smith2009}, e.g., $\tilde{y}_{jt}$ is
unbiased with constant variance, $\sigma^2_j$, and covariance with $\tilde{y}%
_{it}$, $\sigma_{ji}, j,i = 1,2$, the optimal forecast combination weight $%
\eta$ is obtained by solving $\min_{\eta\in[0,1]}\mathbb{E}\left\{y_t-\tilde{%
y}_t^\eta\right\}^2$ and is given by $\eta^\star=[\sigma^2_2-\sigma_{12}]/(%
\sigma^2_1+\sigma^2_2-2\sigma_{12}). $

Following notational conventions (see, e.g., \citealp{West1996}): based on a
sample of size $T+1$, we split this sample into $R$ in-sample observations
used for model fitting, and $P$ out-of-sample observations used for forecast
evaluation, so that $R+P=T+1$. Under this evaluation regime, \cite{Smith2009}
use the out-of-sample MSFE, 
\begin{equation*}
\sigma^2_\eta=P^{-1}\sum_{t=R+1}^{T+1}(y_t-\tilde{y}^\eta_{t})^2,
\end{equation*}
to compare two different forecast combinations: 1) the equally-weighted
combination, with $\eta=1/2$; 2) the sample estimate, $\tilde\eta$, of the
optimal combination, $\eta^\star$. The key finding of \cite{Smith2009} is
that, under their assumptions, when $\eta^\star=1/2$, i.e. when the optimal
approach (`in population') is to actually equally weight the two forecasts, 
\begin{equation*}
\sigma_{\eta=1/2}^2-\sigma_{\tilde\eta}^2\approx -P^{-1}\sum_{t=R+1}^{T+1}(%
\tilde{y}^{\tilde\eta}_t-\tilde{y}^{\eta=1/2}_t)^2\le0.
\end{equation*}
That is, when the optimal combination weight is $\eta^\star=1/2$ the
additional finite-sample noise introduced via the estimation of $\eta$
induces an additional penalty due to estimation error which, in turn,
results in the fixed-weight scheme displaying superior performance. This
finding leads \cite{Smith2009} to conclude that ``The parameter estimation
effect [of the weights] is not large, nevertheless it explains the forecast
combination puzzle.''

\subsection{Extending the Findings of \protect\cite{Smith2009}}

\label{subsec:extsmith}

While the analysis of \cite{Smith2009} is insightful, their findings are not
immediately generalizable to related situations. For instance, the analysis
is based on linear combinations of point forecasts, with unbiased
constituent forecasts that do not depend on unknown parameters.\footnote{%
It is perhaps more accurate to say that the impact of having to estimate
unknown parameters in the constituent models does not feature in their
analysis.} In addition, and most importantly, the optimal MSFE weight is
assumed to be $\eta^\star=1/2$, which coincides with the default
(equally-weighted) combination that typically underpins the forecast
combination puzzle. Finally, this analysis does not immediately extend to
other loss functions.

However, there is now mounting evidence to suggest that the forecast
combination puzzle extends beyond this stylized setup. To this end, we
explore scenarios in which we allow for a range of values for $\eta^{\star}$%
, and a range of values for the fixed weight $\eta$ that \textit{differ}
from $\eta^{\star}$, for the case of both the point forecast combination of 
\cite{Smith2009} \textit{and} a distributional forecast combination, and
where these forecast combinations are estimated using different loss
functions.

Following Section 3.1 of \cite{Smith2009}, consider that the true DGP is
from the AR(2) family, 
\begin{equation*}
y_t = \phi_1 y_{t-1} + \phi_2 y_{t-2} + \epsilon_t,\ \epsilon_t \overset{%
i.i.d.}{\sim} N(0, \sigma^2_{\epsilon}),
\end{equation*}
and that our forecasts are based on a linear pool $f^{(t)}$ of two normal
constituent distributional forecasts $f^{(t)}_1$ and $f^{(t)}_2$: 
\begin{align*}
f^{(t)}_{1}(y) &= N\{y;\gamma_1 y_{t-1},1\}, \\
f^{(t)}_{2}(y) &= N\{y; \gamma_2 y_{t-2},1\}, \\
f^{(t)}(y) &= \eta f^{(t)}_{1}(y) + (1 - \eta) f^{(t)}_{2}(y),
\end{align*}
where $N\{x;\mu,\Sigma\}$ denotes the normal pdf evaluated at $x$ with mean $%
\mu$ and variance $\Sigma$, $\gamma_1$ and $\gamma_2$ are the parameters of
the constituent models, and $\eta$ is the weight assigned to the first model.

We will estimate the parameters of two forecast combinations. First, we
consider the distributional forecast combination given by the linear pool
above, and estimate the parameters ($\eta$, $\gamma_1$, $\gamma_2$) by
minimizing the log loss (or equivalently by maximizing the log likelihood).
Second, we consider the point forecast combination given by the \textit{%
expectation} of the distributional forecast combination, and estimate the
parameters by minimizing the MSFE. The second combination is identical to
the combination considered by \cite{Smith2009} and discussed in the previous
section. Parameters are estimated in the standard two-step fashion, so that $%
\gamma_1$ and $\gamma_2$ are chosen to minimize the selected loss of the
first and second constituent models, respectively, and $\eta$ is estimated
by minimizing the loss of the combination given the aforementioned estimates
for $\gamma_1$ and $\gamma_2$. Section \ref{subsec:defoptcomb} discusses
such optimal forecast combinations in more detail.

For a given (point or distributional) combination, the parameters $%
\phi_1,\phi_2$ and $\sigma^2_\epsilon$ can be chosen so that a desired value
of $\eta^{\star}$ is achieved (we refer to Appendix \ref{subsubsec:dgpparams}
for details). We then test the null of \textit{no inferior forecast accuracy}
across a variety of benchmark forecasts constructed using a range of fixed
weights $\eta \in \{0.25, 0.5, 0.75\}$ (which includes the equally-weighted
benchmark $\eta = 0.5$) against the alternative optimally-weighted
combination.

Figure \ref{fig:smith} presents the rejection frequency ($y$-axis) of the
test across each fixed weight (rows) in the case of the MSFE (for the point
forecast combination, left-hand column) and the log loss (for the
distributional forecast combination, right-hand column). We present these
results for DGP parameter values corresponding to a variety of pseudo-true
weights $\eta^{\star} \in \{0, 0.25, 0.5, 0.75, 1\}$ (colors) and across a
grid of values for the sample size $T+1$ ($x$-axis), comprising the $R =
(T+1)/2$ in-sample observations followed by $P = (T+1)/2$ out-of-sample
observations. For the sake of brevity, we leave a more detailed discussion
of the specific implementation details for this exercise to Appendix \ref%
{subsec:nidextsmith}.

\begin{figure*}[t]
\includegraphics[width = \textwidth]{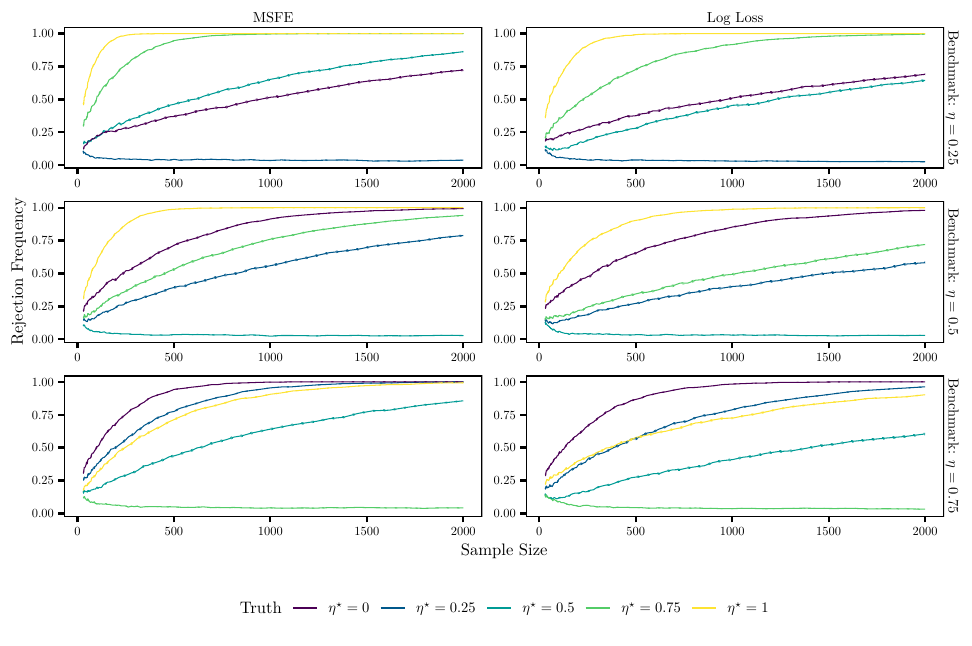}
\caption{Estimates (solid) and their 95\% confidence intervals (dotted) of
the Rejection Frequency ($y$-axis) for the hypothesis test of no inferior
predictive accuracy of a benchmark forecast combination with fixed weights
(rows) against the alternative optimal forecast combination. The test is
conducted with observations drawn from DGPs across a range of pseudo-true
weights (colors), and across a grid of sample sizes ($x$-axis). Results for
a point forecast combination with optimal weights minimizing the MSFE are
given in the first column, and results for a distributional forecast
combination with optimal weights minimizing the log loss are given in the
second column.}
\label{fig:smith}
\end{figure*}

Analyzing each figure we see that, across both loss functions, the
probability of rejection under the null hypothesis of no inferior forecast
accuracy of the benchmark (that is, where the $\eta ^{\star }$ value of the
color equals the $\eta $ value of the row) is (virtually) zero, and
certainly less than 0.05, the nominal size of the test. Analyzing the
rejection rates {in all panels, which are given by the cases where the $\eta 
$ values of the color are different from the $\eta ^{\star }$ value of the
row,} suggests that the forecast combination puzzle persists even when the
distance $\lvert \eta ^{\star }-\eta \rvert $ between $\eta ^{\star }$ and
the incorrect fixed weight $\eta $ is large. For example, a test of no
inferior predictive accuracy of the equally-weighted combination (middle
row) against the optimally-weighted combination where the truth is $\eta
^{\star }=0.25$ (dark blue) has a rejection frequency smaller than 50\%
across all sample sizes less than 1000, whether we are using the MSFE
(left-hand panel) or the log loss (right-hand panel). Hence, the power of
such tests may be quite low in practice even if there are meaningful
differences between forecasts. Later, in Section \ref{subsec:smithone}, we
will conduct a similar exercise to illustrate how this issue of low power
can be resolved by producing optimally-weighted forecast combinations in a
one-step fashion. The production of one-step combinations is described in
the next section.

%In the remainder of this paper, we disentangle the forecast combination puzzle in general settings based on general loss functions, and to understand the importance role played by estimated constituent model parameters.

\section{Measuring the Accuracy of `Optimal' Forecast Combinations}

\label{sec:acc}

Now using more formal notation, let $Y_1, \cdots, Y_{T},\dots$ denote a
sequence of random variables generated from the probability triple $(\mathsf{%
Y}, \mathscr{A}, G)$, where $\mathsf{Y}\subseteq\mathbb{R}^d$, $d\ge1$.
Since the true measure $G$ is unknown in general, we postulate a class of
models $\mathcal{Q}$ on $\mathsf{Y}$, which we identify by their
distribution functions $Q\in\mathcal{Q}$.

Given an observed sample, $y_1,\dots,y_T$, our goal is to predict some
feature of $Y_{T+h}$ at a forecast horizon of $h\ge1$. Since the true model
is unknown, there is no single direction of truth with which to predict
features of interest. The literature has generally maintained that in such
cases entertaining multiple models is a valid approach that can produce
reliable predictions, with the most common approach being to produce
forecast combinations, as we have highlighted.

\subsection{Defining `Optimal' Forecast Combinations}

\label{subsec:defoptcomb}

Denote by $F_{\gamma_1}:\mathsf{Y}\times\Gamma_1\rightarrow C[0,1]$ a
probability measure on $(\mathsf{Y},\mathscr{A})$ indexed by the parameter $%
\gamma_1\in\Gamma_1\subseteq\mathbb{R}^{d_{\gamma_1}}$ and itself lying in
the family $\mathcal{F}(\Gamma_1):=\{F_{\gamma_1}:\gamma_1\in\Gamma_1\}$.
For any $1\le n\le T$, let $\Omega_{n}$ denote the information available to
the forecaster at time $n$, and denote the predictive measure based on time-$%
n$ information as $F_{\gamma_1}^{(n)}:=F_{\gamma_1}(\cdot:\Omega_n)$.

In the majority of forecasting settings, the practitioner entertains a
collection of $K<T$ possible statistical models that can each describe, with
varying accuracy, the movements of the stochastic process $\{Y_t:t\le T\}$.
We consider that each model is specified using a (semi-) parametric family $%
\mathcal{F}(\Gamma_j)$, which depends on $\Gamma_j\subseteq\mathbb{R}%
^{d_{\gamma_j}}$ unknown parameters for each $j=1,\dots,K$. For 
\begin{equation*}
\mathcal{M}=\bigtimes_{j=1}^{K}\mathcal{F}(\Gamma_j),\;\Gamma_j\subseteq%
\mathbb{R}^{d_{\gamma_j}},\text{ for each }j=1,\dots,K,
\end{equation*}%
denoting the collection of all $K$ constituent models, we can combine these
models to produce a forecast combination. To this end, and following \cite%
{Gneiting2013}, we consider the combination function 
\begin{equation*}
C_{\eta}:\mathsf{Y}\times\mathcal{M}\times \mathcal{E}\rightarrow
C[0,1],\quad (F_{\gamma_1},\cdots,F_{\gamma_K})\mapsto
C_\eta(F_{\gamma_1},\cdots,F_{\gamma_K}),\quad \eta\in\mathcal{E}\subseteq%
\mathbb{R}^{d_\eta}.
\end{equation*}
Common choices for the combination family $\mathcal{C}:=\{C_\eta:\eta\in%
\mathcal{E}\}$ include the linear pool (see, e.g., \citealp{Geweke2011}),
and the Beta-transformed linear pool (\citealp{Gneiting2013}).

Given the family of combination functions, $\mathcal{C}$, and the member
model family $\mathcal{M}$, the class of probability measures used for
prediction is the composition of the two: define $\gamma:=(\gamma_1^{\prime
},\dots,\gamma_K^{\prime })^{\prime }$, $\theta:=(\eta^{\prime
},\gamma^{\prime })^{\prime }$, $\Theta:=\mathcal{E}\times\Gamma_1\times%
\dots\times\Gamma_K$ and consider $Q_\theta:\mathsf{Y}\times\mathcal{C}\times%
\mathcal{M}\rightarrow C[0,1]$ defined by $Q_\theta=C_\eta\circ M_\gamma$.
We denote by $\mathcal{Q}$ the class $\{Q_\theta:\theta\in\Theta\}$.%
%Throughout the remainder, which is an element of the class $$\mathcal{Q}:=\{Q_\theta:Q_\eta\in\mathcal{Q},\;(F_{\gamma_1},\cdots,F_{\gamma_{K}})\in\mathcal{M},\;\theta\in\Theta\}.$$

%\subsection{Production of Combination Forecasts }
Generally, the parameters of $Q_\theta$ are unknown and must be estimated.
Throughout, we consider that the forecaster wishes to obtain `optimal'
forecasts in the spirit of \cite{Gneiting2007}, \cite{Gneiting2011}, and 
\cite{Martin2021}. Herein, we take optimal to mean that the distributions we
choose to take out-of-sample are produced by targeting a loss function that
measures precisely the features of $Y_{T+h}$ that are of interest. Following 
\cite{Gneiting2011}, we consider a decision-theoretic framework for such
forecasts. Let $\mathsf{Y}\subseteq\mathbb{R}^d$ denote the \textit{%
observation domain}, and for some $k\ge1$ let $\mathsf{A}\subseteq\mathbb{R}%
^k$ denote the \textit{action space}. We consider two mutually exclusive
cases: the forecaster is interested in measuring predictive accuracy for a
functional of the distribution of $Y_{T+h}$, or the entire distribution.

\subsubsection{Consistent Scoring Functions}

Recall that $\mathcal{Q}$ is a class of distributions on $\mathsf{Y}$, and
consider a functional $U:\mathcal{Q}\mapsto\mathsf{A},\quad Q\mapsto
U[Q]\subseteq \mathsf{A}$ that maps a distribution $Q\in\mathcal{Q}$ to a
subset $U[Q]$ of the action space. A \textit{scoring function} is a
measurable map $S:\mathsf{A}\times \mathsf{Y}\rightarrow[0,\infty)$. We
orient the scoring functions so that a lower score implies a more accurate
forecast. The scoring function $S(\cdot,\cdot)$ \textit{is $\mathcal{Q}$%
-consistent} for a functional $U[\cdot]$ if 
\begin{flalign*}
\text{for all }x\in\mathsf{A},\;Q\in\mathcal{Q}:\quad\mathbb{E}_Q\left\{S\left(U[Q],Y\right)\right\}\leq \mathbb{E}_Q\left[S(x,Y)\right],
\end{flalign*}where $Y$ is a random variable with distribution $Q$, and here
and throughout we assume that any stated expectation exists and is finite.
We say that $S(\cdot,\cdot)$ is \textit{$\mathcal{Q}$-strictly consistent}
for $T[\cdot]$ if 
\begin{equation*}
\mathbb{E}_Q\left\{S\left(U[Q],Y\right)\right\}= \mathbb{E}_Q\left[S(x,Y)%
\right]\implies\;x=U[Q].
\end{equation*}%
If $U[\cdot]$ admits a strictly consistent scoring function, then it is
called \textit{elicitable}. Here and throughout, we consider that the
functional we are interested in is \textit{elicitable}. For several
functionals such as mean, quantiles and expectiles, this is the case, but
there are many functionals, such as variance or expected shortfall, that are
not elicitable, {at least on their own. Some of these functionals are
elicitable jointly with others, which is the case for the pairs (mean,
variance) and (value at risk, expected shortfall), for example (%
\citealp{Fissler2016})}.

Assume the forecaster is interested in the case where $U[\cdot]$ is a
particular functional of the distribution of the random variable $Y_{T+h}$,
and we have available information $\Omega_T$. Then, we are interested in
generating predictions at time $T+h$, $h\ge1$, for the functional $%
U[F^{(T)}_{Y_{T+h}}]$, where we recall that $F_{Y_{T+h}}^{(T)}$ signifies
the distribution of the random variable $Y_{T+h}$ conditional on information 
$\Omega_T$. To this end, we follow, among others, \cite{Patton2020} and
assume that any particular combination model $Q_\theta\in\mathcal{Q}$ admits
a model for $U[F_{Y_{T+h}}^{(T)}]$ of the form $m(Z_T;\theta)=U[Q^{(T)}_%
\theta]$, where $Z_T\in\Omega_T$ denotes observable variables in the
conditioning set and $m:\Omega_T\times\Theta\rightarrow\mathsf{A}$ is known
up to the unknown $\theta$. We can then produce point forecasts for $%
U[F_{Y_{T+h}}^{(T)}]$ using $m(Z_T;\theta)$, and by replacing the unknown $%
\theta$ with 
\begin{equation*}
\widehat\theta_T:=\operatornamewithlimits{argmin\,}_{\theta\in\Theta}%
\sum_{t=1}^{T}S\left[m(Z_t;\theta),y_{t+1}\right].
\end{equation*}

\subsubsection{Proper Scoring Rules}

\label{subsubsec:propriety} The above approach allows one to produce
distributions that generate reliable `point forecasts' for a given functional%
{, but there is no reason for these distributions to be accurate in any
other respect. In cases where we want distributional forecasts that are
accurate as complete representations of the uncertainty surrounding an
unobserved random variable, we can estimate our parameters according a
proper scoring rule.}

A negatively-orientated \textit{proper scoring rule} is a function $S:%
\mathcal{Q}\times\mathsf{Y}\mapsto\mathbb{R}$ such that 
\begin{equation*}
\mathbb{E}_G\left[S(G,Y)\right]\le \mathbb{E}_G\left[S(Q,Y)\right],
\end{equation*}
for all $Q,G\in\mathcal{Q}$. A \textit{strictly} proper scoring rule is a
proper scoring rule that is minimized by $G$ alone.\footnote{%
Scoring rules have a deep connection to decision theory, and we do not
review this literature here (see, e.g., \citealp{Pesaran2002} and %
\citealp{Granger2006} for a discussion in the context of economics and
financial forecasting).} While scoring rules can be used to measure the
accuracy of the predictive distribution or densities, they can also be used
to measure the accuracy of certain features of the distribution: e.g.,
quantiles, or predictive intervals (see, \citealp{Gneiting2007}). Throughout
the remainder, when we refer to scoring rules, it is meant that the action
space is either the class of densities/distributions and not some functional
of it such as quantiles or intervals; we keep this distinction as accuracy
for such latter quantities can be readily measured using (consistent)
scoring functions.

Similar to the case of scoring functions, producing density forecasts
requires estimating $\theta$ in the combination model $Q_\theta\in\mathcal{Q}
$. In this case, we can define an estimator of this parameter by minimizing
the expected scoring rule over our observed sample 
\begin{equation*}
\widehat{\theta}_T:=\operatornamewithlimits{argmin\,}_{\theta\in\Theta}%
\sum_{t=1}^{T}S[Q_\theta^{(t)},y_{t+1}].
\end{equation*}
%The result of this procedure is an optimal forecast distribution ${Q}^{(T)}_{\widehat{\theta}_T}$, where optimal is taken to mean wrt the chosen scoring rule $R(\cdot,\cdot)$

\subsubsection{Estimating Forecast Combinations}

\label{subsubsec:estforecomb}

Regardless of whether one is producing `point' forecasts of a functional or
distributional forecasts, producing the `optimal predictive' is
categorically the same. Therefore, to simplify the presentation, we jointly
treat both approaches. Given a collection of realized observations $\{y_t:
1\le t\le n\}$, $n\le T$, we search for the most accurate combination
predictive distribution $Q_\theta\in\mathcal{Q}$ under the given loss
function $L:\mathcal{Q}\times\mathsf{Y}\rightarrow \mathbb{R}$.

The optimal predictive can generally be produced in two possible ways. In
cases where it is feasible to estimate the parameters jointly, the
combination predictive can be produced by estimating the unknown model
parameters, 
\begin{equation}
\widehat\theta_n:=\operatornamewithlimits{argmin\,}_{\theta\in\theta}L_n(%
\theta) \equiv\operatornamewithlimits{argmin\,}_{\theta\in\theta}%
\sum_{t=1}^{n}\ell_{t+1}(\theta) ,\text{ where }\ell_{t+1}(%
\theta):=L[Q^{(t)}_\theta,y_{t+1}],  \label{eq:1step}
\end{equation}%
and the predictive distribution at time $T+h$ can be taken as $%
Q^{(T)}_{\widehat\theta_n}$. However, due to the dimensionality of $\theta$,
it is often difficult to estimate $\theta$ jointly. Generally, forecast
combinations are estimated in two steps: first, the constituent model
forecasts are produced, and then the combination weights are estimated
conditional on the constituent model forecasts; see, e.g., \cite{Hall2007,
Geweke2011, Gneiting2013}. Throughout we assume that such two-step forecast
combinations are carried out in the spirit proposed in \cite{Gneiting2007}.
Namely, for each constituent model we estimate $\gamma_j$, $j = 1, 2,
\cdots, K, $ via the estimator 
\begin{equation}  \label{eq:gamma}
\tilde{\gamma}_{jn} = \operatornamewithlimits{argmin\,}_{\gamma_j \in
\Gamma_j} \sum_{t=1}^{n}L[Q^{(t)}_\theta,y_{t+1}].
\end{equation}
Collecting the $\tilde\gamma_{jn}$ into $\tilde\gamma_n:=(\tilde%
\gamma_{1n}^{\prime },\dots,\tilde\gamma_{Kn}^{\prime })^{\prime }$, we can
then estimate the combination parameters $\eta$ by maximizing $%
Q_\theta\mapsto L_n(\theta)$, conditional on $\gamma=\tilde\gamma_n$, which
yields 
\begin{equation}  \label{eq:est2s}
\tilde{\theta}_{n}=(\tilde{\eta}_n^{\prime },\tilde{\gamma}_n^{\prime
})^{\prime }= \operatornamewithlimits{argmin\,}_{\theta \in \Theta}
L_n(\eta,\gamma)\text{ s.t. }\gamma=\tilde\gamma_n.
\end{equation}
Once $\tilde\theta_n$ have been calculated the predictive distribution $%
Q^{(T)}_{\tilde\theta_n}$ can be used to produce predictions at time $T+h$.

Throughout the remainder, we refer to the predictive distributions $%
Q^{(T)}_{\widehat\theta_n}$ and $Q^{(T)}_{\tilde\theta_n}$ as the one-step
and two-step predictive combinations, respectively.

Under standard regularity conditions, the extremum estimators $\widehat{%
\theta }_{n}$ and $\tilde{\theta}_{n}$ will converge to well-defined
probability limits, which we denote by $\theta ^{0}:=(\eta ^{0^{\prime
}},\gamma ^{0^{\prime }})^{\prime }$, and $\theta ^{\star }:=(\eta ^{\star
^{\prime }},\gamma ^{\star ^{\prime }})^{\prime }$, respectively;\footnote{%
For more precise definitions of $\theta ^{0}$ and $\theta ^{\star }$, we
refer the interested reader to Appendix \ref{app:A}.} i.e., under our
regularity conditions we will have 
\begin{equation}
\widehat{\theta }_{n}\rightarrow _{p}\theta ^{0}:=(\eta ^{0^{\prime
}},\gamma ^{0^{\prime }})^{\prime }\text{ and }\tilde{\theta}_{n}\rightarrow
_{p}\theta ^{\star }:=(\eta ^{\star ^{\prime }},\gamma ^{\star ^{\prime
}})^{\prime }.  \label{limits}
\end{equation}%
Further, it is well-known, see \cite{Newey1994} for details, that one-step
and two-step estimators do not coincide in general, so that throughout we
can assume that $\theta ^{0}\neq \theta ^{\star }.$ If there is no
asymptotic difference between the one- and two-step estimators, there will
be no asymptotic difference between the resulting forecast combinations
based on these estimators. In most empirical settings, however, significant
differences between one- and two-step estimators exist, and so we restrict
our attention to this case in following analysis.

For a general discussion on two-step estimation see \cite{Newey1994}; for a
more modern treatment see \cite{Frazier2017}; and for a particular approach
to measuring the impact of two-step estimation in the context of
distributional forecast combinations see \cite{Zischke2022}.

\subsection{Testing the Accuracy of Forecast Combinations}

\subsubsection{Evaluation Scheme}

Following \cite{West1996}, \cite{White2000}, and many others, we measure
out-of-sample predictive accuracy using loss differences of forecasts over a
given out-of-sample period. For simplicity, and for consistency with the
earlier exposition of the analysis of \cite{Smith2009}, we let $h=1$ denote
the horizon over which we will make predictions, but note that our results
can also accommodate $h\ge1$ at the cost of additional notation. For ease of
exposition, we re-introduce the notation introduced in Section \ref%
{subsec:revisitingsmith}, and assume the sample consists of $T+1$ total
observations, which we partition into $R$ in-sample periods, and $P$
out-of-sample periods, across which we evaluate the predictions, where $%
R+P=T+1$.

In the approach of \cite{West1996}, the information used to estimate $\theta$
is increased by one unit for each prediction; i.e., at time $t=R+1$, $R$
observations are used to estimate $\theta$, at time $t=R+2$, $R+1$
observations are used, and so forth. This formulation is useful as it allows
one to update the parameter estimates as new information becomes available.
However, the resulting out-of-sample loss difference at time $T+1$ is then a
complex combination of all previous estimators, and also has variability due
to the $P$ out-of-sample observations themselves. Consequently,
disentangling variability due to parameter uncertainty, from the innate
variability of the average loss difference becomes difficult, and obtaining
clear intuition regarding the contribution to each of these pieces to the
behavior of the out-of-sample average loss difference becomes difficult.

Furthermore, it is critical for us to understand the precise impact of
parameter uncertainty on forecast accuracy since this effect, while not
large, ``explains the puzzle'' according to \citet{Smith2009}. Therefore, we
consider a simple framework that cleanly dissects the two types of sampling
variability, loss and parameter estimation. This is accomplished by
estimating the unknown parameters once using $R$ observations, with the
resulting estimators then held fixed over the $P$ out-of-sample periods;%
\footnote{%
More specifically, instead of the first prediction being based on an
estimator $\theta_{R}$, obtained using observations $1,\dots,R$, and the
next based on $\theta_{R+1}$, and so forth, we only consider estimators
based on $R$ observations (with $R\rightarrow\infty$ as $T\rightarrow\infty$%
).} the $P$ out-of-sample periods are then used for evaluation only;
further, we will also maintain that the in-sample period, $R$, and the
out-of-sample period, $P$, are in rough proportion.\footnote{%
We refer to Section 8 of \cite{clark2013advances} for a discussion on the
benefits and disadvantages of various splitting schemes for forecast
evaluation.}

\begin{assumption}[Maintained Assumption]
\label{ass:samp} $R,P\rightarrow\infty$ as $T\rightarrow\infty$, and $%
c:=\lim_TR/P$, with $0<c<\infty$.
\end{assumption}

\subsubsection{Tests of Forecasting Accuracy}

\label{subsubsec:tests}

Following \cite{White2000} and \cite{Hansen2005}, we measure the accuracy of
forecasts by testing the null hypothesis of \textit{no inferior forecast
performance} using the average loss difference over the $P$ out-of-sample
periods. Consider that we wish to test the accuracy of a benchmark forecast
distribution $Q_{\vartheta ^{b}}^{(T)}$, indexed by unknown parameters $%
\vartheta ^{b}$, against an alternative distribution $Q_{\theta ^{a}}^{(T)}$%
, indexed by unknown parameters $\theta ^{a}$. The null hypothesis of 
\textit{no inferior forecast accuracy} of the benchmark ($Q_{\vartheta
^{b}}^{(T)}$) over the alternative ($Q_{\theta ^{a}}^{(T)}$) is 
\begin{equation}
\text{H}_{0}:\;\mathbb{E}(L[Q_{\vartheta ^{b}}^{(T)},Y_{T+1}])\leq \mathbb{E}%
(L[Q_{\theta ^{a}}^{(T)},Y_{T+1}]).  \label{eq:null}
\end{equation}

To test the null in \eqref{eq:null}, we approximate the above expectation,
and the unknown $\vartheta^b,\theta^a$, using their sample counterparts. For
any $t> R$, and any sequence of consistent estimators $\vartheta_R,\theta_R%
\in\Theta$ of $\vartheta^b, \theta^a$, respectively, define the average loss
difference statistic 
\begin{equation*}
\Delta_P(\vartheta_R,\theta_R):=P^{-1}\sum_{t=R+1}^{T+1}d_{t}(\vartheta_R,%
\theta_R), \quad
d_{t}(\vartheta_R,\theta_R)=\ell_{t}(\vartheta_{R})-\ell_{t}(\theta_R),
\end{equation*}
where we recall that $\ell_{t}(\theta_{R})= L[Q^{(t-1)}_{\theta_{R}},y_{t}]$%
. The null hypothesis in \eqref{eq:null} can then be tested using the
standardized statistic 
\begin{equation}
D_{P}(\vartheta_R,\theta_R):=\hat{\Omega}_R^{-1/2}\sqrt{P}%
\Delta_P(\vartheta_R,\theta_R),  \label{dp}
\end{equation}
where $\hat{\Omega}_R$ is a consistent estimator of the asymptotic variance
of $\Delta_P(\vartheta^b,\theta^a)$.

\section{Solving the Forecasting Combination Puzzle}

\label{sec:puzzle}

\subsection{A General Phenomenon}

Our results in Section \ref{subsec:extsmith} demonstrate that even when the
optimal combination weight (in population), $\eta ^{\star }$, is
very different from a fixed, hypothetical combination weight, the standard
approach to testing for differences in forecasting accuracy does not result
in meaningful rejection rates; i.e., even though the hypothetical
combination weight is inferior to the optimal combination weight, the
resulting testing procedure does not reliably detect differences. This
finding holds across two particular loss functions and a host of different
values of $\eta ^{\star }$. In this section, we demonstrate that this
{phenomenon} is present in any class of forecast {combinations produced} in the
standard manner, i.e., in two steps. Thus we give, for the first time, a
truly generic explanation for the puzzle that is agnostic to the chosen
loss, and which is valid under standard regularity conditions.

To state this result, consider the setting where we are given \textit{a
known, i.e., hypothesised} combination weight $\eta _{T}^{\delta }$, which
we can always represent as 
\begin{equation*}
\eta _{T}^{\delta }:=\eta ^{\star }+\delta _{T},
\end{equation*}%
where $\eta ^{\star }$ {is as defined in (\ref{limits}), and }$\eta ^{\star
} $ and $\delta _{T}$ are individually unknown. Consider that the (possibly
random) sequence $\{\delta _{T}:T\geq 1\}$ is described by one of the
following cases: throughout, let $\delta \in \mathcal{E}\subset \mathbb{R}%
^{d_{\eta }}$ be {some non-zero} and bounded vector,

\begin{equation}
\delta _{T}\asymp 
\begin{cases}
\delta /T^{\xi }, & \xi \in \lbrack 0,1/4) \\ 
\delta /T^{\xi }, & \xi =1/4 \\ 
\delta /T^{\xi }, & \xi \in (1/4,\infty )%
\end{cases}%
.  \label{eq:delta}
\end{equation}%
The above class of sequences will allow us to evaluate the behavior of the
standard testing framework for accessing accuracy of different forecast
combination methods across a wide range of hypothesized combination weights $%
\eta _{T}^{\delta }$. In particular, the case $\xi =0$ yields fixed
alternatives, while $\xi =1/2$ yields the class of {canonical Pitman}
sequences. As we shall see, the behavior of the test depends crucially {on
the category within which} $\delta _{T}$ resides.

To this end, consider that our benchmark forecast is $Q_{\vartheta
^{b}}^{(T)}$ {with} $\vartheta ^{b}=(\eta _{T}^{\delta },\gamma ^{\star })$,
and we wish to test the null hypothesis that this benchmark forecast has 
\textit{no inferior forecast performance} relative to the alternative
forecast $Q_{\theta ^{a}}^{(T)}$ {with} $\theta ^{a}=(\eta ^{\star },\gamma
^{\star })$. The null hypothesis in \eqref{eq:null} then becomes 
\begin{equation*}
\text{H}_{0}:\;\mathbb{E}(L[Q_{(\eta _{T}^{\delta },\gamma ^{\star
})}^{(T)},Y_{T+1}])\leq \mathbb{E}(L[Q_{(\eta ^{\star },\gamma ^{\star
})}^{(T)},Y_{T+1}]).
\end{equation*}%
{The null hypothesis} $\text{H}_{0}$ {is then tested using the statistic }in %
\eqref{dp}, {where the} infeasible $\theta ^{a}$ is replaced {by the
feasible estimator,} $\tilde{\theta}_{R}=(\tilde{\eta}_{R}^{\prime },\tilde{%
\gamma}_{R}^{\prime })^{\prime }$,{\ and} the infeasible $\vartheta ^{b}$ by
its feasible counterpart $\vartheta _{R}^{\delta }=(\eta _{R}^{\delta },%
\tilde{\gamma}_{R}^{\prime })^{\prime }$, {where }$\eta _{R}^{\delta }$ {%
varies according to the neighbourhood of }$\eta ^{\star }$ {defined \ in %
\eqref{eq:delta}. }The rejection region for the test is then defined by the
corresponding rejection region 
\begin{equation}
W_{P}(\alpha )=\{D_{P}:D_{P}(\vartheta _{R}^{\delta },\tilde{\theta}%
_{R})>\Phi ^{-1}(1-\alpha )\},\quad \alpha \in (0,1),
\label{eqn:rejectionregion}
\end{equation}%
where $\Phi ^{-1}(\alpha )$ denotes the $\alpha $-quantile of the standard
normal distribution. The {following theorem describes} the behavior of the
resulting test.

\begin{theorem}
\label{lemm:2sla} Let Assumptions \ref{ass:parm}-\ref{ass:quad}, in Appendix %
\ref{app:A} be satisfied. \medskip

\noindent (i) If $\delta_T\asymp \delta/T^\xi$, with $\xi\in[0,1/4)$, and $%
|\delta^{\prime }\nabla_\eta\mathcal{L}(\eta^\star+\delta,\gamma^\star)|>0$,
then $\lim_P\mathrm{Pr}\{W_P(\alpha)\}=1$. \newline

\noindent (ii) If $\delta_T\asymp \delta/T^{1/4}$, then $\lim_P\mathrm{Pr}%
\{W_P(\alpha)\}>\alpha$ or $\lim_P\mathrm{Pr}\{W_P(\alpha)\}<\alpha$,
depending on $\delta$. \newline

\noindent (iii) If $\delta_T\asymp \delta/T^\xi$, with $\xi\in(1/4,\infty]$,
then $\lim_{T}\mathrm{Pr}\{W_P(\alpha)\}=0$ for all $\alpha\in(0,1)$. 
%\newline
\end{theorem}

The above result implies that if two forecasts have combination weights that
are at least $O(T^{-1/4+\varepsilon })$ apart, for $\varepsilon >0$, then
the standard testing approach can distinguish between the forecasts.
Surprisingly, however, if the combination weights are at a distance of $%
CT^{-1/4}$, the test can be arbitrarily over- or under-sized depending on
the magnitude of the sequence $\{\delta _{T}:T\geq 1\}$. More surprisingly,
if two sets of combination weights are within a distance of $%
O(T^{-1/4-\varepsilon })$ from one another, e.g., a parametric neighbourhood
of width $O(T^{-1/2})$, then the test has no power to detect differences
between the combination forecasts. Note that this third result {also encompasses}
the case considered in the illustrative example in Section \ref%
{subsec:revisitingsmith}, in which $\eta _{T}^{\delta }=\eta ^{\star }$ --
i.e. the benchmark fixed weight (denoted by $\eta $ therein) {coincided with}
$\eta ^{\star }$-- and the test displayed zero empirical size (up to Monte
Carlo error).

As an example of {the phenomenon in Theorem \ref{lemm:2sla}(iii)}, consider that we have two competing
combination forecasts defined by different combination weight schemes such
that $\tilde{\eta}_{1}$ and $\tilde{\eta}_{2}$ have distinct asymptotic
distributions, but $R^{1/2}(\tilde{\eta}_{1}-\tilde{\eta}_{2})=O_{p}(1)$,
then the usual test of the null hypothesis of \textit{no inferior predictive
accuracy} will only detect differences between the two forecasts on very
rare occasions, and will detect no statistically significant differences
between the forecasts with probability converging to one.\footnote{%
We recall that under the maintained assumption on $R,P,T$, we have that $%
R\asymp P\asymp T$, so that $O(R^{1/2}/T^{1/2})=O(1)$.}

Practically speaking, Theorem \ref{lemm:2sla} demonstrates that even if the
benchmark forecast, e.g., the equally weighted forecast, is far away from
the optimally weighted combination forecast, then the standard testing
approach is unlikely to reject the inadequacy of this benchmark (with
probability converging to one). In particular, Theorem \ref{lemm:2sla}(i)
demonstrates that the standard test will asymptotically reject the null only
when $\sqrt{T}(\eta^\delta_T-\tilde\eta_R)$ diverges faster than $O(T^{1/4})$,
i.e., when $\xi<1/4$. Therefore, for all intents and purposes, the standard
approach to testing for differences in combination forecasts cannot be
trusted to deliver reliable conclusions in the majority of empirical
situations where it is applied.

\begin{remark}
\label{rem:2} Part (ii) of Theorem \ref{lemm:2sla} is non-standard: in
classical hypothesis testing, sequences of alternatives do not in general
yield a consistent test, but they display at least \textit{some power}
against such a hypothesis. The lack of power in this case is entirely a
consequence of the two-step nature with which the combinations are produced.
\end{remark}

\begin{remark}
It is possible to obtain results similar to those described above in the
original framework of \cite{West1996}. However, the results are not as
intuitive as those presented above, and are much more cumbersome to dissect
and interpret. Therefore, we adopt a fixed-windows estimation scheme to {err}
on the side of simplicity and interpretability instead of technicality.
\end{remark}

%An important implication of Theorem \ref{lemm:2sla} is that the standard critical value used to conduct tests of forecast combinations, $\Phi^{-1}(1-\alpha)$, does not deliver an asymptotically size-controlled test. This demonstrates that, at a minimum, the asymptotic behavior of the test statistic under the null is non-Gaussian. We elaborate on the behavior of the test statistic in the following section. 

\subsection{The Root of the Puzzle}

\label{sec:root}

Theorem \ref{lemm:2sla} demonstrates that when comparing between forecast
combinations, tests of forecast accuracy behave in non-standard ways.
However, it is important to understand the mechanism causing this behavior.
Recall that standard forecast combinations are produced in two steps: first,
we estimate the unknown model parameters $\gamma$ via $\tilde\gamma_R$, then
the combination weights are estimated via $\tilde\eta_R$. In this section,
we show that the two-step nature by which the forecast combination $%
Q^{(T)}_\theta$ is produced results in an average loss difference whose
limiting distribution is non-standard under a large class of combination
weights.

To state the asymptotic distribution of the usual loss-difference test
statistic $\Delta_P(\vartheta^b,\theta^a)$, we require the following
definitions. Let $\mathcal{L}(\theta):=\operatornamewithlimits{plim\,}%
_{P\rightarrow\infty} L_P(\theta)/P$, $\mathcal{M}_{\eta\eta}=
\nabla_{\eta\eta}\mathcal{L}(\eta^\star,\gamma^\star)$, $\mathcal{M}%
_{\gamma\eta}= \nabla_{\gamma\eta}\mathcal{L}(\eta^\star,\gamma^\star)$, and
let $V_{P,R}:=-\mathcal{M}_{\eta\eta}^{-1/2}\{\nabla_\eta
L_P(\eta^\star,\gamma^\star)/P+\mathcal{M}_{\eta\gamma}(\tilde\gamma_R-%
\gamma^\star)\}$. We note that the above exist under Assumptions \ref%
{ass:parm}-\ref{ass:quad} in Appendix \ref{app:A}, and that $V_{P,R}=O_p(1/%
\sqrt{P})$. %For a positive sequence $a_n\rightarrow+\infty$, we say
%that $X_n=\Omega_p(a_n)$ if there exists some positive and non-zero constant 
%$C$, and an $n_0$ such that for all $n\ge n_0$, $X_n\ge Ca_n$, with
%probability converging to one,.

\begin{lemma}
\label{lem:quads} Let Assumptions \ref{ass:parm}-\ref{ass:quad} in Appendix %
\ref{app:A} be satisfied. Recall that $\eta^\delta_T=\eta^\star+\delta_T$.
If $\delta _{T}= \delta /T^{\xi }$, with $\xi \in (0,\infty]$, then, for $%
\vartheta _{R}^{\delta }=(\eta _{T}^{\delta },\tilde{\gamma}_{R})$, 
\begin{equation}
\Delta _{P}(\vartheta _{R}^{\delta },\tilde{\theta}_{R})=\frac{1}{2}%
\|V_{P,R}\|^2-\frac{1}{2}\|\mathcal{J}^{1/2}\left(\eta^\delta_T-\eta^\star%
\right)-\mathcal{J}^{1/2}\left(\tilde\eta_R-\eta^\star\right)-V_{P,R}%
\|^2+o_{p}(\Vert (\eta _{T}^{\delta }-\eta ^{\star })\Vert ^{2}\vee
\|\tilde\eta_R-\eta^\star\|).  \label{p1}
\end{equation}%
%
%
%If $\delta _{T}\asymp \delta /T^{\gamma }$, with $\gamma >1/2$, then, 
%\begin{equation}
% \Delta _{P}(\vartheta _{R}^{\delta },\tilde{\theta}_{R})=\frac{1}{2}%
%\Vert  V_{P}\Vert _{\mathcal{M}_{\eta \eta }}^{2}-\frac{1}{2}\Vert  (\tilde{%
%\eta}_{R}-\eta ^{\star })- V_{P}\Vert _{\mathcal{M}_{\eta \eta
%}}^{2}+o_{p}(\Vert  (\eta _{T}^{\delta }-\eta ^{\star })\Vert ^{2}).
%\label{p2}
%\end{equation}
\end{lemma}

The expansion in Lemma \ref{lem:quads} clarifies the mechanism behind the
behavior exhibited in Theorem \ref{lemm:2sla}. Recall that the behavior of
the standard test is driven by the behavior of $\sqrt{P}\cdot \Delta
_{P}(\vartheta _{R}^{\delta },\tilde{\theta}_{R})$, see equation \eqref{dp}.
However, Lemma \ref{lem:quads} makes clear that if the sequence $%
\{\delta_T:T\ge1\}$ goes to zero fast enough, the limit distribution of $%
\sqrt{P}\cdot\Delta _{P}(\vartheta _{R}^{\delta },\tilde{\theta}_{R})$ is
degenerate. That is, since $V_{P,R}=O(1/\sqrt{P})$ and $(\tilde\eta_R-\eta^%
\star)=O_p(1/\sqrt{P})$ under Assumptions \ref{ass:parm}-\ref{ass:quad},
scaling $\Delta _{P}(\vartheta _{R}^{\delta },\tilde{\theta}_{R})$ by $\sqrt{%
P}$ results in a degenerate test statistic unless 
\begin{equation*}
\operatornamewithlimits{plim\,}_{P\rightarrow\infty}\sqrt{P}%
\|(\eta^\delta_T-\eta^\star)\|^2=\operatornamewithlimits{plim\,}%
_{P\rightarrow\infty}\sqrt{P}\| \delta_T \|^2>0.
\end{equation*}

When $\delta_T=\delta/T^\xi$, $\xi\in(1/4,\infty]$, we have that $%
\operatornamewithlimits{plim\,}_P\sqrt{P}\|(\eta^\delta_T-\eta^\star)\|^2=0$%
, which yields the result in part (iii) of Theorem \ref{lemm:2sla}. If
instead we have $\delta _{T}= T^{-1/4}$, the behavior of $\sqrt{P}\Delta
_{P}(\vartheta _{R}^{\delta },\tilde{\theta}_{R})$ is driven by the
magnitude of 
\begin{equation*}
\frac{1}{2}P^{1/4}(\eta ^{\star }-\eta _{T}^{\delta })^{\prime }\mathcal{M}%
_{\eta \eta }P^{1/4}(\eta ^{\star }-\eta _{T}^{\delta })=\delta ^{\prime }%
\mathcal{M}_{\eta \eta }\delta +o_{p}(1),
\end{equation*}%
which yields the second result in Theorem \ref{lemm:2sla}. The above term
also drives the power of the statistic $D_P(\vartheta_{R},\tilde\theta_R)$
in the case where $\delta_T\in[0,1/4)$ (i.e., part (i) of Theorem \ref%
{lemm:2sla}).\footnote{%
We note, however, that Lemma \ref{lemm:2sla} is not valid, as stated, in the
case where $\delta_T\in[0,1/4)$ since the remainder term in \eqref{p1} is no
longer negligible, since $\sqrt{P}\|\eta^\delta_T-\eta^\star\|^2$ will
diverge. Nonetheless, a similar argument to that used to prove Lemma \ref%
{lem:quads} can be used to deal with this case.}

To obtain the limit distribution of the out-of-sample average loss
difference under the null hypothesis in \eqref{eq:null}, based on the
benchmark forecast combination $Q^{(T)}_{\vartheta^{(b)}}$ with $%
\vartheta^{(b)}=(\eta^\delta_T,\gamma^\star)$, and with $\eta^\delta_T$ as
in \eqref{eq:delta}, we require a few additional definitions. Let $%
X_{P}:=\nabla_\eta L_P(\eta^\star,\gamma^\star)/P$, $Z_{R,\gamma}:=(\tilde%
\gamma_P-\gamma^\star)$, so that we can write $V_{P,R}=-\mathcal{M}%
_{\eta\eta}^{-1/2}(X_P+\mathcal{M}_{\eta\gamma}Z_{R,\gamma})$. Under
Assumptions \ref{ass:parm}-\ref{ass:quad} in Appendix \ref{app:A}, we have
that $\sqrt{P} X_{P}\Rightarrow X\sim N(0,\Sigma_X)$, where $\Sigma_X:=\lim_P%
\text{Var}\{\nabla_\eta L_P(\eta^\star,\gamma^\star)/\sqrt{P}\}$ and $\sqrt{R%
} Z_{R,\gamma}\Rightarrow Z_\gamma\sim N(0,\Sigma_\gamma)$, and where 
\begin{equation*}
\Sigma_Z:=\lim_{R\rightarrow\infty}\text{Var}[\{\left[\nabla_{\gamma\gamma}%
\mathcal{L}(\gamma^\star_1)\right]^{-1}\nabla_{\gamma_1} L_R(\gamma^\star_1)/%
\sqrt{R},\dots,\left[\nabla_{\gamma\gamma}\mathcal{L}(\gamma^\star_1)\right]%
^{-1}\nabla_{\gamma_K} L_R(\gamma^\star_K)/\sqrt{R}\}^{\prime }].
\end{equation*}%
Recall that $c:=\lim_T R/P$, with $0<c<\infty$, and by Assumption \ref%
{ass:dist}, for some matrix $Q_V$, we have that $\mathcal{V}:=(\mathcal{V}%
_1^{\prime },\mathcal{V}_2^{\prime })^{\prime }=\mathcal{L}_{\infty}(c\sqrt{R%
}V_{R,R},\sqrt{P}V_{P,R})=N(0,Q_V)$, for some matrix $Q_V$, is the limit
(joint) law of the terms $c\sqrt{R}V_{R,R}$ and $\sqrt{P}V_{P,R}$.

\begin{corollary}
\label{cor:quads} Under Assumptions \ref{ass:parm}-\ref{ass:dist}, the
statistic $P\cdot \Delta_{P}(\vartheta^\delta_R,\tilde\theta_R)$ has the
following behavior.

\noindent (i) If $\delta_T\asymp \delta/T^\xi$, with $\xi\in(0,1/2)$, then $%
P\cdot \Delta_{P}(\vartheta^\delta_R,\tilde\theta_R)\rightarrow+\infty, $
with probability converging to one.

\noindent (ii) If $\delta_T= \delta/T^{1/2}$, and $\theta^\star\in\mathrm{Int%
}(\Theta)$, then 
\begin{equation*}
P\cdot \Delta_{P}(\vartheta^\delta_R,\tilde\theta_R)\Rightarrow \frac{1}{2}%
\|X+c^{-1}\mathcal{M}_{\eta\gamma}Z_\gamma\|_{\mathcal{M}_{\eta\eta}^{-1}}^2-%
\frac{1}{2}\|\left\{1/(1+c)\right\}^{1/2}\delta-(\mathcal{V}_1-\mathcal{V}%
_2)\|^2. 
\end{equation*}

\noindent (iii) If $\delta_T= \delta/T^\xi$, with $\xi\in(1/2,\infty]$, and
if $\theta^\star\in\mathrm{Int}(\Theta)$, then 
\begin{equation*}
P\cdot \Delta_{P}(\vartheta^\delta_R,\tilde\theta_R)\Rightarrow \frac{1}{2}%
\|X+c^{-1}\mathcal{M}_{\eta\gamma}Z_\gamma\|_{\mathcal{M}_{\eta\eta}^{-1}}^2-%
\frac{1}{2}\|\mathcal{V}_1-\mathcal{V}_2\|^2. 
\end{equation*}
\end{corollary}

When $\delta_T\asymp \delta/T^\xi$ and $\xi\in(0,1/2)$, the result of Corollary \ref
{cor:quads}(i) demonstrates that under the null hypothesis in \eqref{eq:null} with
benchmark forecast combination $Q_{\vartheta ^{b}}^{(T)}$, $\vartheta
_{}^{(b)}=(\eta _{T}^{\delta },{\gamma}^{\star})$, and alternative forecast combination $
Q_{\theta ^{a}}^{(T)}$, with $\theta ^{a}=\theta^\star$, the test statistic $%
P\cdot \Delta _{P}(\vartheta _{R}^{\delta },\tilde{\theta}_{R})$, where $%
\vartheta _{R}^{\delta }=(\eta _{T}^{\delta },\tilde{\gamma}_{R})$,
diverges. Hence, if $\eta^\star$ and $\eta^\delta_T$ are sufficiently different, we can accurately learn differences between competing forecast combination methods. In contrast, when $\delta_T\asymp \delta/T^\xi$
, with $\xi\in[1/2,\infty]$, Corollary \ref{cor:quads}(ii)-(iii)
demonstrates that the asymptotic distribution of $P\cdot \Delta
_{P}(\vartheta _{R}^{\delta },\tilde{\theta}_{R})$ is non-standard.

In the regime where $\delta_T\asymp \delta/T^\xi$, with $\xi\in[1/2,\infty]$%
, the test statistic $P\cdot \Delta _{P}(\vartheta _{R}^{\delta },\tilde{%
\theta}_{R})$ {converges} in distribution to a random variable with two
components. {The first component is a generalized}
chi-squared random variable, which does not admit a closed-form formula for
its density or distribution function.\footnote{%
Since the quadratic form $\Vert X+\mathcal{M}_{\eta \gamma }Z_\gamma\Vert _{%
\mathcal{M}_{\eta \eta }^{-1}}^2$ cannot be re-written as a quadratic form 
with an idempotent weighting matrix, the distribution is not chi-squared.}
{The second component is itself possibly comprised of two components:} so long as $\xi\ge1/2$, the second term depends on the difference of
two mean-zero but correlated normal {random variables}, which captures the behavior
of (a scaled version of) the out-of-sample loss difference due to
differences in the combination weights, i.e., the $\eta$-components, and a centering term that captures the difference between $\eta^\delta_T$ and $\eta^\star$.\footnote{%
	This term results from taking a {second-order} Taylor expansion of the loss
	difference, and grouping terms appropriately.}
 In the regime where $\xi=1/2$, an additional term is present that captures the fact that $\sqrt{P}(\eta^\delta_T-\eta_\star)=(\sqrt{P/T})\sqrt{T}(\eta^\delta_T-\eta^\star)=\sqrt{P/(P+R)}\delta$, which, since $c=\lim_T R/P$ converges to $(1/(1+c))^{1/2}\delta$. 

However, it is important to note that since the first {component} of the asymptotic distribution is a generalized chi-squared random variable, even if the
second {component were} not present it would be infeasible to obtain
closed-form quantiles for the null distribution of the test statistic. In
addition, it is important to realise that the distribution of this test
statistic depends on the loss used in the analysis, the choice of
constituent models, and the specific combination function chosen. As such,
there is no{ hope for a single set of generally applicable critical values}, and any
simulated critical values will need to be application-specific {as a consequence.}

While it is not feasible to obtain closed-form quantiles for the null
distribution of the test statistic, since the second component is always
negative, it is possible to deduce a conservative test that uses just the
critical values of the generalized chi-distribution. That is, since 
\begin{equation*}
\frac{1}{2}\|X+c^{-1}\mathcal{M}_{\eta\gamma}Z_\gamma\|_{\mathcal{M}%
_{\eta\eta}^{-1}}^2-\frac{1}{2}\|\mathcal{V}_1-\mathcal{V}_2\|^2\le \frac{1}{%
2}\|X+c^{-1}\mathcal{M}_{\eta\gamma}Z_\gamma\|_{\mathcal{M}%
_{\eta\eta}^{-1}}^2 
\end{equation*}%
we can use the quantiles of the generalized chi-distribution to deduce a
conservative, but feasible, test of the the null that a benchmark forecast $%
Q^{(T)}_{\vartheta^b}$, with $\vartheta^{(b)}=(\eta^\delta_T,\gamma^\star)$,
is not inferior to an alternative forecast $Q^{(T)}_{\theta^a}$, with $%
\theta^a=\theta^\star$. Such a test will be asymptotically conservative in general, but
will not be too conservative  so long as the differences between $\sqrt{R}%
V_{R,R}$ and $\sqrt{P}V_{P,R}$ are small.

To implement such a test, for a {given} benchmark forecast $%
Q_{\vartheta^{(b)}}^{(T)}$ based on known combination weight $\eta^\delta_T$%
, we must construct a feasible critical value from the generalized
chi-distribution $\frac{1}{2}\|X+c^{-1}\mathcal{M}_{\eta\gamma}Z_\gamma\|_{%
\mathcal{M}_{\eta\eta}^{-1}}^2$ via simulation. In particular, we start out
by first simulating $h=1,\dots,H$ realisations for the random variables $%
X^{(h)},Z^{(h)}$ from normal distributions where $\Sigma_X,\Sigma_Z$ are
replaced with consistent estimators based on $\eta^\delta_R,\tilde\gamma_R$,
and then, for each $h=1,\dots,H$, we form the statistic 
\begin{equation*}
\Delta^{(h)}:=\frac{1}{2}\|X^{(h)}+(P/R)^{1/2}\widehat{\mathcal{M}}_{\eta
\gamma}Z^{(h)^{\prime }}\|^2_{\widehat{\mathcal{M}}_{\eta\eta}^{-1}},
\end{equation*}%
where $\widehat{\mathcal{M}}_{\eta\eta}$ and $\widehat{\mathcal{M}}%
_{\gamma\eta}$ are the usual sample estimate counterparts of the matrices $%
\mathcal{M}_{\eta\eta}$ and $\mathcal{M}_{\gamma\eta}$ and calculated at $%
\eta^\delta_R,\tilde\gamma_R$. Sorting $\{\Delta^{(h)}:h=1,\dots,H\}$, we
can obtain an $\alpha$-level critical value by choosing the $%
\lceil(1-\alpha)H \rceil$-th smallest value, and rejecting the null when the
observed value of the statistic $P\cdot
\Delta_{P}(\vartheta^\delta_R,\tilde\theta_R)$ exceeds this value.

Denoting the above simulated critical value by $\widehat{\mathrm{cv}}%
_{\lceil(1-\alpha)H\rceil}$, we can define a test of the null hypothesis of
no inferior predictive accuracy of the benchmark model $Q^{(T)}_{\eta^%
\delta_T,\gamma^\star}$ against the alternative model $Q^{(T)}_{\theta^\star}
$ that is estimated in two-steps, via the corresponding rejection region 
\begin{equation}
{W}^{2s}_{P}(\alpha )=\{\Delta_{P}:P\cdot\Delta_{P}(\vartheta _{R}^{\delta },%
\tilde{\theta}_{R})>\widehat{{\mathrm{cv}}}_{\lceil(1-\alpha)H\rceil}\},%
\quad \alpha \in (0,1).  \label{eqn:2srejectionregion}
\end{equation}%
Unlike the usual test, based on $W_{P}(\alpha)$, the test based on the
rejection region ${W}^{2s}_{P}(\alpha )$ explicitly accounts for the
two-step nature by which $P\cdot\Delta_{P}(\vartheta _{R}^{\delta },\tilde{%
\theta}_{R})$ is constructed. As we have already seen, failure to account
for the two-step nature of the forecast combinations does not deliver a test
with appropriate size control. Conversely, since a test based on ${W}%
^{2s}_{P}(\alpha )$ accounts for the two-step nature of the estimation, the
test should deliver appropriate size control.

We note that, while it is feasible to use the above simulation method to
construct an appropriate critical value for the testing differences between
competing forecast combinations, we do not necessarily advocate for this
approach in general. In particular, while the above procedure would deliver
an appropriately sized test, it is unclear if the above test is the most
powerful approach, and additional research is necessary to determine this.
In addition, as we elaborate on in the following Section \ref{subsec:avoid},
in certain cases there is a simpler alternative to {adopting} the above simulated
critical value that allows us to entirely avoid the forecast combination
puzzle.

\subsection{Revisiting the \protect\cite{Smith2009} Two-Step Combinations}

{Meanwhile, in this section}, we return to the example setting discussed in \cite%
{Smith2009} and demonstrate {numerically} that the simulated critical {value}
suggested in Section \ref{sec:root} delivers a testing procedure with
approximately correct size under the null hypothesis of no inferior forecast
accuracy, and also has higher power than the standard approach under the
alternative.

To demonstrate {the empirical size of the test}, we return to the example in Section \ref{sec:motiv}
and consider values of $\phi_1,\phi_2$ and $\sigma^2$ for the AR(2) family
such that, under MSFE (respectively, log-loss), the optimal value of the
combination weight, $\eta^\star$, obtained under MSFE (respectively,
log-loss), is (approximately) equal to the benchmark combination weight of
1/2, i.e., the equally-weighted combination benchmark.\footnote{%
Under log score, a DGP that {ensures that} the optimal two-step forecast combination
weight is (approximately) equal to 1/2 can be obtained by setting $%
(\phi_1,\phi_2,\sigma^2)=(.4000,-.4421,1)^\top$ in the AR(2) model. For
MSFE, taking $(\phi_1,\phi_2,\sigma^2)=(.4000,-.4070,1)^\top$ delivers a DGP
such that the optimal (two-step) forecast combination weight is
(approximately) 1/2. The values of $(\phi_1,\phi_2,\sigma^2)$ in both cases
were obtained numerically by maximizing the corresponding loss function
using a sample size of {10} million observations generated from the DGP.}
Under this DGP, we generate 1000 replications across three different samples
sizes $T=1000,2000$ and $5000$. Across each of these datasets, we use the
first $R=T/2$ observations for training, and the remaining $P=T-R$
observations for testing.

The first three rows of Table \ref{tab:reg} compare the size of the standard
test of no inferior forecast accuracy, based on $W_P(\alpha)$ in %
\eqref{eqn:rejectionregion}, against the testing approach that caters for
the two-step nature of the combination forecast construction, via the use of 
${W}^{2s}_{P}(\alpha )$ in \eqref{eqn:2srejectionregion}. In the two-step
approach, we use $B=10000$ draws to simulate the critical value in all
cases. As a comparator, we also present the empirical rejection frequency of
the standard {`t-test' of} the null hypothesis that $\eta=1/2$,
{based on} the correct two-step standard error for $\tilde\eta_R$. This
additional {comparator serves as a benchmark of sorts, enabling us to {diagnose}, in some sense, if the rejection
rates for the tests are due primarily to differences in the predictive
ability of the combinations, or are due to differences in the combination
weights themselves.} 

The {results} in the table demonstrate that, as already highlighted in Section %
\ref{sec:motiv}, the standard approach has zero size in (nearly) all cases
when the benchmark combination weight, 1/2 in this case, is close to the
optimal combination weight. In contrast, even with an asymptotically
conservative critical value, accounting for the two-step nature by which the
forecasts were produced, via ${W}^{2s}_{P}(\alpha )$ in %
\eqref{eqn:2srejectionregion}, results in a test that has sizes {that} are
much closer to the nominal level. Depending on the chosen loss, {for the finite samples used, }the test
based on ${W}^{2s}_{P}(\alpha )$ can be slightly {over- or under-sized}, but
delivers results that are much closer to the nominal level than {does} the standard
approach.

The second set of three rows in Table \ref{tab:reg}, compares the {empirical} power of
the {alternative} testing approaches, under both loss functions, {under} a DGP that is a
relatively small distance away from the DGP that delivers {an optimal
combination weight that is} equal to the benchmark combination weight of
1/2.\footnote{%
Under log-loss, the DGP generating the observed data is fixed at $%
(\phi_1,\phi_2,\sigma^2)=(.40,-.50,1)^\top$, while under MSFE we take the
DGP generating the data to be $(\phi_1,\phi_2,\sigma^2)=(.40,-.45,1)^\top$.
That is, in both cases, the DGP deviates from the version that delivered
equivalence between the equal-weighted combination and the
optimally-weighted combination by moving the second root of the AR(2)
process, $\phi_2$, by -.05.}Across both loss functions, our testing approach
has vastly higher power than the standard approach, {with the standard testing
approach based on the log-loss having zero power under this minor deviation from 
the null hypothesis, and with its power under the MSFE being only around one
percent. Moreover, the magnitude of the power of the correctly sized test is broadly similar (for any given sample size) under the two losses. Interestingly enough, the comparator `t-test' has quite high power under log loss, but very different, and much lower power under MSFE.} 

Consequently, the results in Table \ref{tab:reg} demonstrate empirically
that the if one wishes to conduct a test of forecast accuracy based on
approaches that use forecast combinations, accounting for the two-step
nature by which these forecasts were produced will be critical in producing
tests with good power and reliable size.

\begin{table}[]
\centering
{\footnotesize \ 
\begin{tabular}{lrrrrrr}
\hline\hline
\textbf{Size} & \multicolumn{3}{c}{MSFE} & \multicolumn{3}{c}{Log-Loss} \\ 
& Two-Step & T-Test & Stand & Two-Step & T-Test & Stand \\ \hline
$T=1000$ & 0.0202 & 0.0348 & 0.0000 & 0.0530, & 0.1690 & 0.0000 \\ 
$T=2000$ & 0.0240 & 0.0330 & 0.0004 & 0.0440, & 0.3250 & 0.0000 \\ 
$T=5000$ & 0.0322 & 0.0392 & 0.0000 & 0.0510, & 0.5480 & 0.0000 \\ 
\hline\hline
\textbf{Power} & \multicolumn{3}{c}{MSFE} & \multicolumn{3}{c}{Log-Loss} \\ 
& Two-Step & T-Test & Stand & Two-Step & T-Test & Stand \\ \hline
$T=1000$ & 0.1544 & 0.1426 & 0.0122 & 0.1960 & 0.3960 & 0.0000 \\ 
$T=2000$ & 0.3038 & 0.2536 & 0.0108 & 0.3670 & 0.6980 & 0.0000 \\ 
$T=5000$ & 0.6546 & 0.5450 & 0.0102 & 0.6870 & 0.9570 & 0.0000 \\ 
\hline\hline
\end{tabular}%
}  
\caption{Monte Carlo rejection rates for testing the null hypothesis in 
\eqref{eq:null} under the size and power DGPs for the loss functions MSFE
and Log-Loss. The table compares the rejection rates of the standard test
(Stand) and the two-step version based on the rejection region in 
\eqref{eqn:2srejectionregion} (Two-Step). The column T-Stat gives the
empirical rejection rates of the T-test that $\protect\eta=1/2$. }
\label{tab:reg}
\end{table}

\subsection{Avoiding the Puzzle}

\label{subsec:avoid}

As noted, the use of a two-step forecast combination, i.e., $%
Q^{(T)}_{\tilde\theta_R}$, is much more common in practice than a one-step
forecast combination, i.e., $Q^{(T)}_{\widehat\theta_R}$, where $%
\widehat\theta_R$ is defined in \eqref{eq:1step}. However, the analysis in
Section \ref{sec:puzzle} demonstrates that the forecast combination puzzle
is entirely due to this (two-step) estimation approach. This then begs the
question of whether it may be possible to avoid the puzzle altogether by
changing the way forecast combinations are produced.

In this section we compare the accuracy of the one-step forecast
combination, $Q^{(T)}_{\widehat\theta_R}$ against the standard two-step
forecast combination, $Q^{(T)}_{\tilde\theta_R}$. Since the two-step
approach is the standard approach in the literature, we test that the 
\textit{two-step (benchmark) approach is not inferior to the one-step
(alternative) approach}. The following result demonstrates that if the
one-step forecast combination approach is computationally feasible, it will
always yield superior forecast performance.

\begin{theorem}
\label{thm:diff}  Under Assumptions \ref{ass:parm}-\ref{ass:dist} in
Appendix \ref{app:A}, $\mathrm{Pr}[D_P(\tilde\theta_T,\widehat\theta_T)>0]%
\rightarrow1$ as $n\rightarrow\infty.$
\end{theorem}

Theorem \ref{thm:diff} yields the following immediate corollary on the
forecast accuracy of a benchmark equally-weighted combination against the
optimally estimated one-step combination: let $\theta_R^{ew}=(K^{-1}\iota^{%
\prime },\tilde\gamma_R^{\prime })^{\prime }$, with $\iota$ a $K$%
-dimensional vector of ones.

\begin{corollary}
\label{cor:eqs}  If $\eta^0\ne K^{-1}\iota$, and Assumptions \ref{ass:parm}-%
\ref{ass:dist} in Appendix \ref{app:A} are satisfied, then  $\mathrm{Pr}%
[D_P(\theta^{ew}_R,\widehat\theta_T)>0]\rightarrow1$ as $n\rightarrow\infty.$
\end{corollary}

\begin{remark}
Theorem \ref{thm:diff} and Corollary \ref{cor:eqs} state that when measuring
predictive accuracy using average loss differences, optimally estimated
one-step forecast combinations always (weakly) outperform two-step forecast
combination approaches, including the equally-weighted combination. This
result is loss function agnostic, and applies to any strictly proper scoring
rule or any consistent scoring function.
\end{remark}

\begin{remark}
In certain cases, joint optimization over the parameters in the forecast
combination (i.e., in one-step) will not produce a unique minima for the
optimization program in \eqref{eq:1step}; for instance, in cases where we
wish to combine two mean point forecasts from separate regression models,
then a one-step optimization may yield a set of optimal values. In such
cases, the resulting point onto which the one-step estimator $%
\widehat\theta_n$ converges, i.e., $\theta^0$, will not be unique. However,
if one is only interested in forecasting accuracy, then this is immaterial
for practice: since the forecasts are constructed from $L[Q_%
\theta^{(T)},Y_{T+1}]$, and are evaluated using an out-of-sample estimator
for $\mathbb{E} L[Q_\theta^{(T)},Y_{T+1}]$, each of the optima will
(asymptotically) produce the same level of accuracy for the forecast
combination. Moreover, since the assumptions used to derive the results in
Theorem \ref{thm:diff} do not require point identification, the results on
the accuracy of the one-step forecast combination given above will remain
valid.
\end{remark}

\subsection{Revisiting Smith and Wallis (2009) for One- and Two-Step
Combinations}

\label{subsec:smithone}

In Figure \ref{fig:avoid}, we repeat the simulation of Section \ref%
{subsec:extsmith}, but for different benchmark-alternative pairs (rows),
with a view to illustrating the implications of Theorem \ref{thm:diff} and
Corollary \ref{cor:eqs}, and thereby illustrating the benefits of the
one-step combination.\footnote{{We omit results for $\eta^{\star} = 0$ and $%
\eta^{\star} = 1$ to exclude cases where $\theta^0 = \theta^{\star}$,
whereby the one- and two-step parameter estimators converge to the same
values in the limit, violating Assumption \ref{ass:cons}.}} As before,
discussion of implementation details is left to Appendix \ref%
{subsec:nidpuzzlecause}.

\begin{figure*}[ht]
\includegraphics[width = \textwidth]{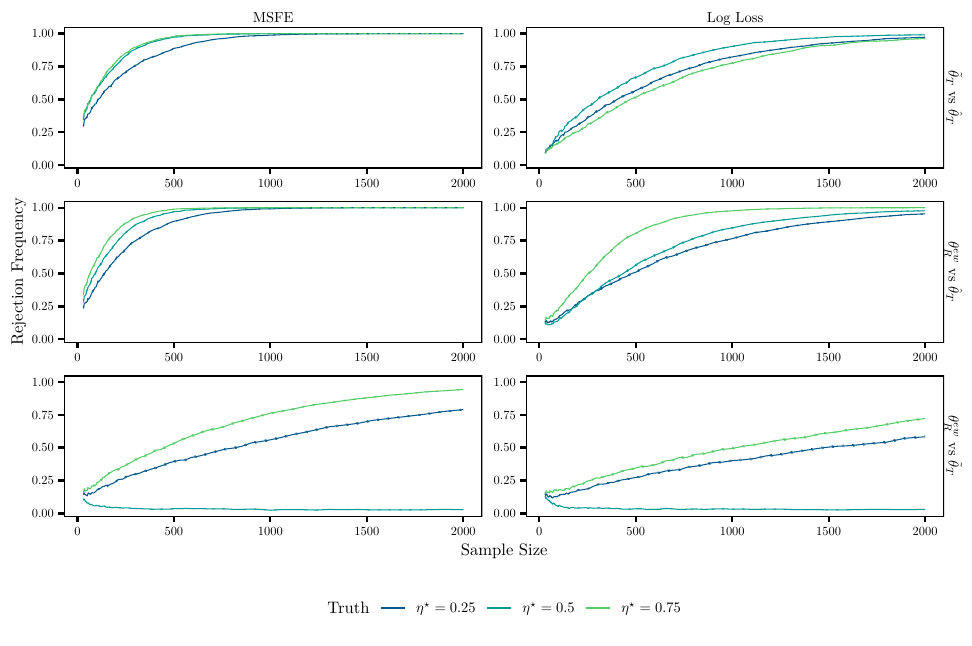}  
\caption{Estimates of the rejection frequency ($y$-axis) for the hypothesis
test of no inferior predictive accuracy of a benchmark forecast combination
against an alternative combination (rows, benchmark vs alternative). 95\%
confidence intervals for the rejection frequency only cover the thickness of
the line, and are omitted. The test is conducted with observations drawn
from DGPs across a range of pseudo-true weights (colours), and across a grid
of sample sizes ($x$-axis). Results for a point forecast combination with
optimal weights minimising the MSFE are given in the first column, and
results for a distributional forecast combination with optimal weights
minimising the log loss are given in the second column.}
\label{fig:avoid}
\end{figure*}

In the first row, we test the benchmark two-step combination against the
alternative one-step combination, and find that the rejection frequency ($y$%
-axis) quickly converges to one (in favour of the one-step combination) as
the sample size ($x$-axis) increases, for combinations optimising both the
MSFE (left-hand column) and the log loss (right-hand column) and for all
values of $\eta^{\star}$ (colours). This reflects the result of Theorem \ref%
{thm:diff}, and supports preferring the one-step combination over its
two-step counterpart when optimising the combination to maximise forecast
performance.

In the second row, the hypothesis of no inferior predictive accuracy of the
equally-weighted two-step benchmark is tested against the alternative
one-step combination. Here, we again find that the rejection frequency ($y$%
-axis) rapidly converges to one as the sample size ($x$-axis) increases, for
DGPs of \textit{all} limiting two-step weights $\eta^{\star}$, including $%
\eta^{\star} = 0.5$, where $\theta^{ew}_R$ is the best-performing two-step
combination. This is due to the fact that the one-step combination $\hat{%
\theta}_T$ converges to a higher-performing combination in the limit than is
possible for \textit{any} two-step combination, including $\theta^{ew}_R$ in
the case where the optimal two-step weight $\eta^{\star} = 0.5$; this
follows since even if $\eta^{\star} = 0.5$ under the two-step approach, the
optimal one-step weight $\eta^0 \neq 0.5$.

Displayed in the bottom row are rejection frequencies for the test of no
inferior predictive accuracy of the equally-weighted two-step benchmark
against the optimally-weighted two-step alternative. Rejection frequencies
for the same test were also displayed in the middle row of Figure \ref%
{fig:smith}, and here we see again that for the two-step combination this
test is undersized and has low power, even when the equally-weighted vector $%
\theta^{ew}_R$ is far from optimal (dark blue and green), leading to the
forecast combination puzzle. Comparing the bottom and middle rows of Figure %
\ref{fig:avoid}, we find that the power of the test increases dramatically
to resolve the puzzle when estimating parameters in one step (middle row)
rather than two (bottom row). The increase in power from one-step estimation
is seen across all sample sizes ($x$-axis), all pseudo-true weights
(colours) and both losses (columns).

\subsection{Revisiting \cite{Geweke2011} for One- and Two-Step
Combinations}

% The following tables were generated by geweke_plot.R and saved as
% figure/TBL3.tex and figure/TBL4.tex, respectively.

\begin{table}[ht]
\centering
\begin{tabular}{lr}
  \hline
Combination & Average Log Score \\ 
  \hline
Equally-Weighted Two Step & 3.3481 \\ 
  Optimally-Weighted Two Step & 3.3459 \\ 
  One Step & 3.3596 \\ 
   \hline
\end{tabular}
\caption{The average log scores of three forecast combinations {for} S\&P500 returns.}
\label{tbl:gewekeaveragelogscores} 
\end{table}

% latex table generated in R 4.3.1 by xtable 1.8-4 package
% Fri Jul 28 15:19:07 2023
\begin{table}[ht]
\centering
\begin{tabular}{llr}
  \hline
Benchmark & Alternative & $p$ value \\ 
  \hline
Equally-Weighted Two Step & Optimally-Weighted Two Step & 0.8251 \\ 
  Equally-Weighted Two Step & One Step & 5.675e-05 \\ 
  Optimally-Weighted Two Step & One Step & 6.935e-12 \\ 
   \hline
\end{tabular}
\caption{{P-values} (right column) for tests of the null hypothesis that a benchmark combination (left column) is not inferior to an alternative combination (right column). All combinations produce one-step-ahead distributional forecasts for S\&P500 returns, and are evaluated on a log-score basis.} 
\label{tbl:geweketests}
\end{table}

In this section, we give an empirical example which demonstrates that one-step combinations resolve the forecast combination puzzle. {Specifically}, we follow Section 3 of \cite{Geweke2011} and consider a linear pool comprising the Gaussian exponential GARCH$(1,1)$ {model} (``EGARCH'') and the GARCH$(1,1)$ model with i.i.d. Student $t$ {errors} (``$t$-GARCH''); this pool is then used to  produce one-step-ahead distributional forecasts of daily {logarithmic} S\&P500 returns. All parameters (including combination weights) are estimated using returns for the 3783 trading days from years 1990 to 2004, inclusive ({the} ``training set''). We then evaluate and compare the log-score-based forecasting performance of different combinations using returns for the out-of-sample period comprising the 3772 trading days from years 2005 to 2019, also inclusive ({the} ``test set'').

Three different ways of estimating {the forecast combinations} are compared: equally-weighted two-step estimation, optimally-weighted two-step estimation, and one-step estimation. In {the} two-step combinations, the EGARCH parameters are first chosen to maximize the {average log score of the EGARCH one-step-ahead predictive distribution over the training set}; likewise for the $t$-GARCH parameter estimates. To produce the equally-weighted two-step combination, we set the weights to $0.5$. The combination {weights for the} optimally-weighted two-step combination are estimated by maximizing the log score of the combination across the training set, with the EGARCH and $t$-GARCH parameters fixed at their first-step values. {The one-step combination parameters are jointly} chosen to maximize the training-set average log score of the one-step-ahead predictive distribution of the combination density, in a single optimization program.

Table \ref{tbl:gewekeaveragelogscores} contains the training-set average log scores of the three combinations (first column). Consistent with the forecast combination puzzle, even across the training-set the equally-weighted two-step combination outperforms the optimally-weighted two-step combination, since it has a higher average log score (second column). As suspected, both two-step combinations are beaten by the one-step approach over the training-set.

In Table \ref{tbl:geweketests} we display the $p$ values (right column) for three tests of the null hypothesis that a benchmark combination (left column) is not inferior to an alternative combination (middle column). The test proceeds according to Section \ref{subsubsec:tests} using the loss differences pertaining to the log scores of the out-of-sample test set. The asymptotic variance of the loss difference is estimated using the method described in Section 4.1 of \cite{Okui2010}, with the quadratic spectrum kernel and $S = \sqrt{T}$. In the first row we fail to reject the null hypothesis that the equally-weighted two-step combination is not inferior to the optimally-weighted two step combination, reflecting the forecast combination puzzle. The tests displayed in the second and third rows unequivocally reject the null that the benchmark - either {the} equally- or {the optimally-weighted two-step combination} - is not inferior to the one-step {combination, with both null hypotheses} rejected at the 1\% level. This perfectly reflects the theoretical {results} in \ref{thm:diff} and Corollary \ref{cor:eqs} in Section \ref{subsec:avoid}: we can avoid the puzzle and obtain a higher performing forecast combination by estimating all parameters in a single step, rather than the standard two steps.

\section{Conclusion}

In this paper, we investigate the forecast combination puzzle through the
lens of hypothesis testing approaches aimed at discriminating between the
relative performances of equally-weighted and optimally-weighted forecast
combinations. Forecast combination parameters are optimized according to a
scoring function or scoring rule for point forecasting and distributional
forecasting, respectively, and we thereby demonstrate that the forecast
combination puzzle is a phenomenon that extends far beyond point forecasts
optimized according to the MSFE.

Our theoretical analysis demonstrates that such hypothesis tests have no
local power, and lack size control, when applied to hypothesis tests aiming
to distinguish between the performance of optimally- and equally-weighted 
\textit{two-step} forecast combinations -- a result that is unusual for
tests of this nature. {This perverse behavior is entirely due to the fact
that the test statistic used to conduct the test does not account for the
two-step nature by which forecast combinations are produced. Consequently,
as we show in the paper, under the null hypothesis the limiting distribution
of the (appropriately scaled) test statistic is not asymptotically normal,
but instead converges in distribution to a generalized chi-distribution.}

The source of this problematic property is that all uncertainty in the
performance of optimally-weighted two-step forecast combinations derives
entirely, at first order, from sampling variability in the parameter
estimates of the constituent models, with no contribution from the
estimation of the weights. It is the relative performance of the different
approaches to obtaining the weights, however, that is the subject of the
hypothesis test -- a subject that does not contribute (asymptotically, at
least) to the sampling variability of the performance measure, resulting in
a test with no local power.

An extension to the Monte Carlo exercise of \cite{Smith2009} illustrates how
this lack of local power can permeate throughout the parameter space. By
producing the rejection frequencies of a variety of hypothesis tests of no
inferior forecast accuracy of a fixed-weight benchmark against the
optimally-weighted two-step alternative, we illustrate that such hypothesis
tests can require large sample sizes to reject in favor of the
optimally-weighted combination, even when the (unknown) best-performing
weights are very different from the vector of equal weights. This finding
was seen repeatedly, for point and distributional forecast combinations
optimized according to different scores, across a diverse range of DGP
parameter values, and across several fixed-weighted (and not just
equally-weighted) benchmark combinations. We also revisit a two-model distributional forecast combination of S\&P500 returns in \cite{Geweke2011} to obtain empirical evidence that this phenomenon occurs in practice.

It is shown that, under mild assumptions, optimizing combination parameters
in one step will always eventually reject the null hypothesis that it does
not have a higher forecast accuracy than an equally-weighted or
optimally-weighted two-step benchmark combination. Repeating the Monte Carlo
exercise for the one-step alternative under the optimally- and
equally-weighted two-step benchmarks reveals that the low power that is
characteristic of the forecast combination puzzle in the case of two-step
alternatives is absent when the alternative optimally-weighted combination
is estimated in one-step. In addition, we verify the superiority of one-step combinations in the  S\&P500 returns example considered in \cite{Geweke2011} by showing that {the one-step density combination delivers superior predictive accuracy relative to two-step benchmarks.}

In this way, we argue that the root-cause behind the lack of evidence for
the performance of optimally-weighted two-step combinations against their
equally-weighted counterparts is an artefact of the way in which such
combinations are generally produced. Consequently, the ``puzzle'' is
evidenced by the low power observed under standard hypothesis testing
approaches used to obtain evidence in favor of, or against,
optimally-weighted combinations produced in the standard two-step manner.
Furthermore, we demonstrate that if it feasible to produce optimal
combinations in a single step, the forecasting puzzle can be completely
avoided. Hence, if the problem at hand is such that forecast combinations
can be produced in a one-step fashion, the practitioner will (always) reap
appreciable gains, in terms of forecast accuracy, by undertaking such a
strategy.

Alternatively, if a one-step approach is infeasible, or if a two-step
approach is simply preferable, we have demonstrates how the usual testing
framework must be altered to accommodate the two-step nature by which the
combinations were produced, and to ensure that the resulting behavior does
not corrupt our testing results. In particular, we have show that under a
broad range of DGPs, when using a two-step forecast combination approach,
both the test statistic and the critical value employed must be altered in
order to deliver a test that has correct size, and meaningful power.

Before concluding, we note that there are many interesting cases where a
forecast combination procedure may seem, at the outset, not to be produced
in a two-step fashion, but which upon closer inspection reveals that such
forecast combinations are actually produced in (at least) a two-step
fashion. As an illustrative example, consider the context of volatility
forecasting using the HAR model (\citealp{corsi2009simple}; %
\citealp{Corsi2012}). By viewing HAR model forecasts as the combination of
lagged moving average models for realised variance, \cite{Clements2021}
document the existence of a forecasting combination puzzle in HAR models and
show that such models do not generally perform better than a simple weighted
average of the constituent forecasts.

Interestingly, the results of \cite{Clements2021} seem to document the
existence of a forecast combination puzzle without a `first-stage'
estimation step being required to produce the forecasts; at face value it
then seems that such an example lies outside the scope of our general
results. However, recall that the ``observed value'' of realised volatility
is not a genuine realisation of ``observed data'', but a nonparametric
estimator of integrated-variance, computed using inter-daily returns. That
is, the very construction of the realised variance series constitutes a
first-stage estimation step, and forecasts produced via HAR models can
therefore be viewed as two-stage forecast combination methods: the first
stage estimates the realised variance series, and the second the combination
scheme. Hence, the lack of power HAR models exhibit to distinguish between
equally and optimally weighted forecast combinations is also explained by
our theoretical results. We leave a full study on such types of first-stage
estimation steps for future research, but remark that, the heavy use of
high-frequency returns, and realised variance in particular, in financial
forecasting applications would seem to imply the existence of undiscovered
combination puzzles.

\bibliographystyle{apalike}
\bibliography{library}

\begin{thebibliography}{}

\bibitem[Aastveit et~al., 2019]{Aastveit2019}
Aastveit, K.~A., Mitchell, J., Ravazzolo, F., and van Dijk, H.~K. (2019).
\newblock The evolution of forecast density combinations in economics.
\newblock In {\em Oxford Research Encyclopedia of Economics and Finance}.
  Oxford University Press.

\bibitem[Andrews, 1999]{andrews1999estimation}
Andrews, D.~W. (1999).
\newblock Estimation when a parameter is on a boundary.
\newblock {\em Econometrica}, 67(6):1341--1383.

\bibitem[Baran and Lerch, 2018]{Baran2018}
Baran, S. and Lerch, S. (2018).
\newblock Combining predictive distributions for the statistical
  post-processing of ensemble forecasts.
\newblock {\em International Journal of Forecasting}, 34(3):477--496.

\bibitem[Bates and Granger, 1969]{Bates1969}
Bates, J.~M. and Granger, C.~W. (1969).
\newblock The combination of forecasts.
\newblock {\em Journal of the Operational Research Society}, 20(4):451--468.

\bibitem[Chan and Pauwels, 2018]{Chan2018}
Chan, F. and Pauwels, L.~L. (2018).
\newblock Some theoretical results on forecast combinations.
\newblock {\em International Journal of Forecasting}, 34(1):64--74.

\bibitem[Claeskens et~al., 2016]{Claeskens2016}
Claeskens, G., Magnus, J.~R., Vasnev, A.~L., and Wang, W. (2016).
\newblock The forecast combination puzzle: A simple theoretical explanation.
\newblock {\em International Journal of Forecasting}, 32(3):754--762.

\bibitem[Clark and McCracken, 2013]{clark2013advances}
Clark, T. and McCracken, M. (2013).
\newblock Advances in forecast evaluation.
\newblock {\em Handbook of economic forecasting}, 2:1107--1201.

\bibitem[Clements and Vasnev, 2021]{Clements2021}
Clements, A. and Vasnev, A.~L. (2021).
\newblock Forecast combination puzzle in the har model.
\newblock {\em Available at SSRN: https://ssrn.com/abstract=3875026 or
  http://dx.doi.org/10.2139/ssrn.3875026}.

\bibitem[Corsi, 2009]{corsi2009simple}
Corsi, F. (2009).
\newblock A simple approximate long-memory model of realized volatility.
\newblock {\em Journal of Financial Econometrics}, 7(2):174--196.

\bibitem[Corsi et~al., 2012]{Corsi2012}
Corsi, F., Audrino, F., and Renò, R. (2012).
\newblock Har modeling for realized volatility forecasting.
\newblock In {\em Handbook of Volatility Models and Their Applications}, pages
  363--382. John Wiley \& Sons, Inc, Hoboken, NJ, USA.

\bibitem[Elliott, 2011]{Elliott2011}
Elliott, G. (2011).
\newblock Averaging and the optimal combination of forecasts.
\newblock {\em Manuscript, Department of Economics, UCSD}.

\bibitem[Fissler and Ziegel, 2016]{Fissler2016}
Fissler, T. and Ziegel, J.~F. (2016).
\newblock Higher order elicitability and osband’s principle.
\newblock {\em The Annals of Statistics}, 44(4):1680--1707.

\bibitem[Frazier and Renault, 2017]{Frazier2017}
Frazier, D.~T. and Renault, E. (2017).
\newblock Efficient two-step estimation via targeting.
\newblock {\em Journal of Econometrics}, 201(2):212--227.

\bibitem[Geweke and Amisano, 2011]{Geweke2011}
Geweke, J. and Amisano, G. (2011).
\newblock Optimal prediction pools.
\newblock {\em Journal of Econometrics}, 164(1):130--141.

\bibitem[Gill et~al., 2021]{Gill2021}
Gill, P.~E., Murray, W., and Wright, M.~H. (2021).
\newblock {\em Numerical linear algebra and optimization}.
\newblock SIAM.

\bibitem[Gneiting and Raftery, 2007]{Gneiting2007}
Gneiting, T. and Raftery, A.~E. (2007).
\newblock Strictly proper scoring rules, prediction, and estimation.
\newblock {\em Journal of the American Statistical Association}.

\bibitem[Gneiting and Ranjan, 2011]{Gneiting2011}
Gneiting, T. and Ranjan, R. (2011).
\newblock Comparing density forecasts using threshold- and quantile-weighted
  scoring rules.
\newblock {\em Journal of Business \& Economic Statistics}, 29(3):411--422.

\bibitem[Gneiting and Ranjan, 2013]{Gneiting2013}
Gneiting, T. and Ranjan, R. (2013).
\newblock Combining predictive distributions.
\newblock {\em Electronic Journal of Statistics}, 7:1747--1782.

\bibitem[Graefe et~al., 2014]{graefe2014combining}
Graefe, A., Armstrong, J.~S., Jones~Jr, R.~J., and Cuz{\'a}n, A.~G. (2014).
\newblock Combining forecasts: An application to elections.
\newblock {\em International Journal of Forecasting}, 30(1):43--54.

\bibitem[Granger and Machina, 2006]{Granger2006}
Granger, C.~W. and Machina, M.~J. (2006).
\newblock Forecasting and decision theory.
\newblock {\em Handbook of Economic Forecasting}, 1:81--98.

\bibitem[Hall and Mitchell, 2007]{Hall2007}
Hall, S.~G. and Mitchell, J. (2007).
\newblock Combining density forecasts.
\newblock {\em International Journal of Forecasting}, 23(1):1--13.

\bibitem[Hansen, 2005]{Hansen2005}
Hansen, P.~R. (2005).
\newblock A test for superior predictive ability.
\newblock {\em Journal of Business \& Economic Statistics}, 23(4):365--380.

\bibitem[Johnson, 2022]{Johnson2022}
Johnson, S.~G. (2022).
\newblock The nlopt nonlinear-optimization package.

\bibitem[Kraft, 1988]{Kraft1988}
Kraft, D. (1988).
\newblock A software package for sequential quadratic programming.
\newblock {\em Forschungsbericht- Deutsche Forschungs- und Versuchsanstalt fur
  Luft- und Raumfahrt}.

\bibitem[Kraft, 1994]{Kraft1994}
Kraft, D. (1994).
\newblock Algorithm 733: Tomp--fortran modules for optimal control
  calculations.
\newblock {\em ACM Transactions on Mathematical Software (TOMS)},
  20(3):262--281.

\bibitem[Makridakis et~al., 2018]{Makridakis2018}
Makridakis, S., Spiliotis, E., and Assimakopoulos, V. (2018).
\newblock The m4 competition: Results, findings, conclusion and way forward.
\newblock {\em International Journal of Forecasting}, 34(4):802--808.

\bibitem[Makridakis et~al., 2020]{Makridakis2020}
Makridakis, S., Spiliotis, E., and Assimakopoulos, V. (2020).
\newblock The m4 competition: 100,000 time series and 61 forecasting methods.
\newblock {\em International Journal of Forecasting}, 36(1):54--74.

\bibitem[Martin et~al., 2021]{Martin2021}
Martin, G.~M., Loaiza-Maya, R., Maneesoonthorn, W., Frazier, D.~T., and
  Ram{\'\i}rez-Hassan, A. (2021).
\newblock Optimal probabilistic forecasts: When do they work?
\newblock {\em International Journal of Forecasting}.

\bibitem[Newey and McFadden, 1994]{Newey1994}
Newey, W.~K. and McFadden, D. (1994).
\newblock Large sample estimation and hypothesis testing.
\newblock {\em Handbook of Econometrics}, 4:2111--2245.

\bibitem[Okui, 2010]{Okui2010}
Okui, R. (2010).
\newblock Asymptotically unbiased estimation of autocovariances and
  autocorrelations with long panel data.
\newblock {\em Econometric Theory}, 26(5):1263--1304.

\bibitem[Opschoor et~al., 2017]{Opschoor2017}
Opschoor, A., Van~Dijk, D., and van~der Wel, M. (2017).
\newblock Combining density forecasts using focused scoring rules.
\newblock {\em Journal of Applied Econometrics}, 32(7):1298--1313.

\bibitem[Patton, 2020]{Patton2020}
Patton, A.~J. (2020).
\newblock Comparing possibly misspecified forecasts.
\newblock {\em Journal of Business \& Economic Statistics}, 38(4):796--809.

\bibitem[Pesaran and Skouras, 2002]{Pesaran2002}
Pesaran, M.~H. and Skouras, S. (2002).
\newblock Decision-based methods for forecast evaluation.
\newblock {\em A Companion to Economic Forecasting}, pages 241--267.

\bibitem[Ranjan and Gneiting, 2010]{Ranjan2010}
Ranjan, R. and Gneiting, T. (2010).
\newblock Combining probability forecasts.
\newblock {\em Journal of the Royal Statistical Society: Series B (Statistical
  Methodology)}, 72(1):71--91.

\bibitem[Satop{\"a}{\"a} et~al., 2014]{Satopaeae2014}
Satop{\"a}{\"a}, V.~A., Baron, J., Foster, D.~P., Mellers, B.~A., Tetlock,
  P.~E., and Ungar, L.~H. (2014).
\newblock Combining multiple probability predictions using a simple logit
  model.
\newblock {\em International Journal of Forecasting}, 30(2):344--356.

\bibitem[Smith and Wallis, 2009]{Smith2009}
Smith, J. and Wallis, K.~F. (2009).
\newblock A simple explanation of the forecast combination puzzle.
\newblock {\em Oxford Bulletin of Economics and Statistics}, 71(3):331--355.

\bibitem[Stock and Watson, 2004]{Stock2004}
Stock, J.~H. and Watson, M.~W. (2004).
\newblock Combination forecasts of output growth in a seven-country data set.
\newblock {\em Journal of Forecasting}, 23(6):405--430.

\bibitem[Stone, 1961]{Stone1961}
Stone, M. (1961).
\newblock The linear opinion pool.
\newblock {\em Ann. Math. Statist}, 32:1339--1342.

\bibitem[Taylor, 2020]{Taylor2020}
Taylor, J.~W. (2020).
\newblock Forecast combinations for value at risk and expected shortfall.
\newblock {\em International Journal of Forecasting}, 36(2):428--441.

\bibitem[Thorey et~al., 2018]{Thorey2018}
Thorey, J., Chaussin, C., and Mallet, V. (2018).
\newblock Ensemble forecast of photovoltaic power with online crps learning.
\newblock {\em International Journal of Forecasting}, 34(4):762--773.

\bibitem[Timmermann, 2006]{Timmermann2006}
Timmermann, A. (2006).
\newblock Forecast combinations.
\newblock {\em Handbook of Economic Forecasting}, 1:135--196.

\bibitem[Van Der~Vaart, 1998]{VanDerVaart1998}
Van Der~Vaart, A. (1998).
\newblock {\em Asymptotic Statistics}.
\newblock Cambridge University Press.

\bibitem[van~der Vaart et~al., 1996]{van1996weak}
van~der Vaart, A., van~der Vaart, A.~W., van~der Vaart, A., and Wellner, J.
  (1996).
\newblock {\em Weak Convergence and Empirical Processes: With Applications to
  Statistics}.
\newblock Springer Science \& Business Media.

\bibitem[Wang et~al., 2018]{Wang2018}
Wang, L., Wang, Z., Qu, H., and Liu, S. (2018).
\newblock Optimal forecast combination based on neural networks for time series
  forecasting.
\newblock {\em Applied Soft Computing}, 66:1--17.

\bibitem[West, 1996]{West1996}
West, K.~D. (1996).
\newblock Asymptotic inference about predictive ability.
\newblock {\em Econometrica: Journal of the Econometric Society}, pages
  1067--1084.

\bibitem[White, 2000]{White2000}
White, H. (2000).
\newblock A reality check for data snooping.
\newblock {\em Econometrica}, 68(5):1097--1126.

\bibitem[Yuan and Jennrich, 1998]{yuan1998asymptotics}
Yuan, K.-H. and Jennrich, R.~I. (1998).
\newblock Asymptotics of estimating equations under natural conditions.
\newblock {\em Journal of Multivariate Analysis}, 65(2):245--260.

\bibitem[Zischke et~al., 2022]{Zischke2022}
Zischke, R., Martin, G.~M., Frazier, D.~T., and Poskitt, D.~S. (2022).
\newblock On measuring the sampling variability of estimated combinations of
  distributional forecasts.
\newblock {\em Unpublished Manuscript}.

\end{thebibliography}

\appendix

\section{Assumptions and Proofs}

\subsection{Assumptions and Discussion}

\label{app:A}

We wish to treat cases where the model combination weights, $\eta$, are
allowed to lie on the boundary of the parameter space, and we wish to be
agnostic about the asymptotic distribution of the estimated parameters in
the constituent models. Before stating our maintained assumptions, we recall
that $L_n(\theta)=\sum_{t=1}^{n-1}\ell_t(\theta)$, for some loss function $%
\ell_{t+1}(\theta)=L[P^{(t)}_\theta,y_{t+1}]$.

\begin{assumption}
\label{ass:parm} The parameter space $\Theta$ is compact, and can be written
as a Cartesian product of intervals with the form $[0,c]$ or $[-c,c]$, for
some $0<c<\infty$ that can change dimension-by-dimension.
\end{assumption}

\begin{remark}
The assumption that $\Theta$ is a Cartesian product is not onerous since the
parameter space is a product of closed intervals, and since
the requirement that the boundary is `on the left', and at zero, is without
loss of generality: if the $j$-th element of the parameter vector originally
satisfies $\theta_j\in[c_{1j},c_{2j}]$, then we can always consider the
translated parameter $\vartheta_j = \theta_j-c_1$, which lies in $%
[0,c_{2j}-c_{1j}]$. For instance, this assumption is immediately satisfied
for combination weights in the case of linear pools. Further, this
assumption satisfies the conditions on the parameter space necessary to
apply the results of \cite{andrews1999estimation}.
\end{remark}

Since our goal is not inference on $\gamma^\star$, we maintain the following
high-level regularity condition for the estimated parameters in the
constituent models.

\begin{assumption}
\label{ass:bound} The population criterion $\sum_{j=1}^{K}\mathcal{L}%
(\gamma_k)$ exists and is uniquely minimized at $\gamma^\star=(\gamma_1^%
\star,\dots,\gamma^\star_K)^{\prime }$; (ii) $\sqrt{n}(\tilde\gamma_n-%
\gamma^\star)=O_p(1)$.
\end{assumption}

Together with Assumption \ref{ass:parm}, the following conditions give
consistency of the one- and two-step estimators $\widehat{\theta}%
_n,\tilde\theta_n$ to their corresponding limit optimizers.

\begin{assumption}
\label{ass:cons} There exists a function $\theta\mapsto\mathcal{L}(\theta)$ such that: (i) $\sup_{\theta\in\Theta}|L_n(\theta)/n-\mathcal{L}%
(\theta)|=o_p(1)$; (ii) There exist a non-empty set $\Theta_{\mathrm{I}%
}\subset\Theta$ with a finite number of elements such that, for each $%
\theta^0\in\Theta_{\mathrm{I}}$, the map $\theta\mapsto \mathcal{L}(\theta)$
is minimized at $\theta^0\in\Theta_{\mathrm{I}}$, i.e., $\Theta_{\mathrm{I}%
}:=\operatornamewithlimits{arginf\,}_{\Theta}\mathcal{L}(\theta)$; $%
\eta\mapsto\mathcal{L}(\eta,\gamma^\star)$ is uniquely minimized at $%
\eta^\star\in\mathcal{E}$, and $\theta^{\star}\notin\Theta_{\mathrm{I}}$.
\end{assumption}

Assumption \ref{ass:cons} differs from the usual point-identification
assumption imposed in classical extremum estimation problems. Instead, we
only require that the one-step estimator is set-identified, and that the set
contains only a finite collection of values. This is helpful for treating
situations where the forecast combination delivers unidentified parameter
estimates due to the existence of multiple roots in the criterion function.
However, from the point of forecasting accuracy whether $\Theta_{\mathrm{I}}$
is a singleton is irrelevant so long as the forecast accuracy associated
with the collection of points in $\Theta_{\mathrm{I}}$ is constant, which is
precisely what Assumption \ref{ass:cons} stipulates.

The following assumption allows us to deduce quadratic expansions that we use to establish the behavior of tests of forecasting accuracy. In what
follows, we recall that $\nabla_\theta \mathcal{L}(\theta)$ denotes
left/right (hereafter, l/r) derivatives or standard derivatives, depending
on the context.

\begin{assumption}
\label{ass:quad} (i) For some $\delta>0$, $\vartheta\in\Theta_{\mathrm{I}%
}\cup\{\theta^\star\}$, and all $\|\theta-\vartheta\|<\delta$, $L_n(\theta)$
and $\mathcal{L}(\theta)$ admit second-order l/r partial derivatives in $%
\theta$ that are continuous with respect to components that can be perturbed
to the l/r; (ii) for each $\vartheta\in\Theta_{\mathrm{I}}$, $\nabla_\theta
L_n(\vartheta)=O_p(1)$ and $[\nabla_\eta
L_n(\eta,\gamma^\star);\nabla_\gamma L_n(\gamma^\star)]/\sqrt{n}=O_p(1)$;
(iii) the matrices $\nabla_{\eta\eta}\mathcal{L}(\eta,\gamma^\star)|_{\eta=%
\eta^\star}$ and $\nabla_{\theta\theta}\mathcal{L}(\theta)|_{\theta=\theta^%
\star}$ are positive-definite; (iv) For any $\delta_n=o(1)$, 
%$\sup_{\|\eta-\eta^\star\|\le\delta_n}\|\nabla_{\eta\eta}L_n(\eta,\gamma^\star)/n-\nabla_{\eta\eta}\mathcal{L}(\eta,\gamma^\star)\|=o_p(1)$,
$\sup_{\vartheta\in\Theta_{\mathrm{I}}}\sup_{\|\theta-\vartheta\|\le%
\delta_n}\|\nabla_{\theta\theta}L_n(\theta)/n-\nabla_{\theta\theta}\mathcal{L%
}(\theta)\|=o_p(1)$.
\end{assumption}

\begin{remark}
The assumptions regarding the continuous l/r derivatives can be relaxed in
cases by replacing the differentiability conditions with a `stochastic
differentiability' condition that depend on the specific type of loss
function, and which are well-known in the literature on empirical processes
(see, e.g., \citealp{van1996weak}). Since this complication is not our main
interest, we maintain the stronger conditions on the l/r derivatives.
\end{remark}

Assumptions \ref{ass:parm}-\ref{ass:quad} are sufficient to deduce
consistency, as well as the rate of convergence, for the one-step estimator
of the parameters in the forecast combination. The following result follows
similar arguments to Theorem 2 \cite{andrews1999estimation}, but requires
slight alternations since we do note assume the existence of a unique
minimum for $\mathcal{L}(\theta)$. In what follows, define $\mathcal{M}%
_{\theta\theta}(\theta):=\nabla_{\theta\theta}\mathcal{L}(\theta)$, with
derivatives for $\eta$ and $\gamma$ defined accordingly.

\begin{lemma}
\label{lem:lem1} Suppose Assumptions \ref{ass:parm}, \ref{ass:cons} and \ref%
{ass:quad} are satisfied, then for some $\vartheta\in\Theta_{\mathrm{I}}$
the following are satisfied.

\begin{enumerate}
\item[(i)] $\sqrt{n}(\widehat\theta_n-\vartheta)=O_p(1)$,

\item[(ii)] For $\lambda_n=\sqrt{n}(\widehat\theta_n-\vartheta)$, and $Z_n:=-%
\mathcal{M}_{\theta\theta}(\vartheta)^{-1}\nabla_\theta L_n(\vartheta)/\sqrt{%
n} $, 
\begin{equation*}
L_n(\widehat\theta_n)-L_n(\vartheta)=-\frac{1}{2}Z_n^{\prime }\mathcal{M}%
_{\theta\theta}(\vartheta)Z_n+\frac{1}{2}(\lambda_n-Z_n)^{\prime }\mathcal{M}%
_{\theta\theta}(\vartheta)(\lambda_n-Z_n)+o_p(1).
\end{equation*}
\end{enumerate}
\end{lemma}

A result like Lemma \ref{lem:lem1} can also be deduced for the two-step
estimator of the combination weights.

\begin{lemma}
\label{lem:quad_exp}Under Assumptions \ref{ass:parm}-\ref{ass:quad},

\begin{enumerate}
\item[(i)] $\sqrt{n}(\tilde\eta_n-\eta^\star)=O_p(1)$.

\item[(ii)] For $\mathcal{J}=\mathcal{M}_{\eta\eta}(\theta^\star)$, $%
\kappa_n=\mathcal{J}^{1/2}\sqrt{n}(\tilde\eta_n-\eta^\star)$, and ${V}_n:=-%
\mathcal{J}^{-1/2}\{\nabla_\eta L_n(\theta^\star)/\sqrt{n}+\mathcal{M}%
_{\eta\gamma}(\theta^\star)\sqrt{n}(\tilde\gamma_n-\gamma^\star)\}$, 
\begin{equation*}
L_n(\tilde\theta_n)-L_n(\eta^\star,\tilde\gamma_n)=-\frac{1}{2}\|{{V}_n}\|^2+%
\frac{1}{2}\|\kappa_n-{V}_n\|^2+o_p(1).
\end{equation*}
\end{enumerate}
\end{lemma}

Lemma \ref{lem:quad_exp} demonstrates that the two-step criterion has a
quadratic expansion that depends on the centering variable $V_n$, which is a
linear function of $\sqrt{n}(\tilde\gamma_n-\gamma^\star)$. This centering
sequence is a consequence of the two-step nature with which the criterion
has been optimized, which treats $\sqrt{n}(\tilde\gamma_n-\gamma^\star)$ as
a fixed quantity. In contrast, a joint expansion of $L_n(\tilde\theta_n)$
about $\theta^\star$ produces a linear expansion at first-order.

\begin{lemma}
\label{lem:lem2}If Assumptions \ref{ass:parm}-\ref{ass:quad} are satisfied,
then 
\begin{equation*}
n^{-1/2}\{L_n(\tilde\eta_n,\tilde\gamma_n)-L_n(\eta^\star,\gamma^\star)\}=%
\nabla_{\gamma}\mathcal{L}(\theta^\star)^{\prime }\sqrt{n}%
(\tilde\gamma_n-\gamma^\star)+o_p(1).
\end{equation*}
\end{lemma}

\begin{remark}
Lemma \ref{lem:lem2} demonstrates that the dominant term in the expansion of 
$\{L_n(\tilde\eta_n,\tilde\gamma_n)-L_n(\eta^\star,\gamma^\star)\}$ is the
linear term $\nabla_{\gamma}\mathcal{L}(\theta^\star)^{\prime }\sqrt{n}%
(\tilde\gamma_n-\gamma^\star)$. That is, the behavior of $\sqrt{n}%
(\tilde\eta_n-\eta^\star)$ is irrelevant in determining the behavior of $%
L_n(\tilde\theta_n)/\sqrt{n}$. Lemma \ref{lem:lem2} is a consequence of the
two-step nature of the estimator $\tilde\theta_n$. 
%Furthermore, as Lemmas \ref{thm:2s} and \ref{lemm:2sla} demonstrate, this result has critical implications for the behavior of forecast accuracy when using forecast combinations (see, in particular, Lemmas \ref{thm:2s} and \ref{lemm:2sla}).	
\end{remark}

\begin{remark}
The above result demonstrates that the estimation, and limiting behavior, of
the combination weights has no impact on $D_P$ defined in \eqref{dp}: the
asymptotic behavior of $D_P$ only depends on the pseudo-true value $%
\eta^\star$, and the variation in $\tilde{\gamma}_R$. So long as $\sqrt{R}
(\tilde\eta_R-\eta^\star)=O_p(1)$, and even if $\eta^\star$ is on the
boundary of the parameter space, the variability in the estimates of the
parameters governing the constituent forecasts ultimately determines the
behavior of $D_P$.
\end{remark}

The following remaining assumptions are used to control the behavior of the
average loss difference.

\begin{assumption}
\label{ass:dist} The variable $Z_n=\left[L_n(\theta^\star)/\sqrt{n}%
-L_n(\theta^0)/\sqrt{n}-\left\{\mathcal{L}(\theta^\star)-\mathcal{L}%
(\theta^0)\right\},\sqrt{n}(\tilde\gamma_n-\gamma^\star)\right]^{\prime },$
is such that $Z_n \Rightarrow N(0,Q)$. Further, for $V_{P,R}:=\nabla_\eta
L_P(\theta^\star)/P+\mathcal{M}_{\eta\gamma}(\tilde\gamma_R-\gamma^\star)$,
we have that $(\sqrt{R}V_{R,R},\sqrt{P}V_{P,R})\Rightarrow N(0,Q_V)$ for
some positive semi-definite covariance matrix $Q_V$.
\end{assumption}

%\noindent The following assumption is used to link the in-sample dataset used to fit the parameters, to the out-of-sample set used for evaluation.

\subsection{Proofs of Main Results}

\label{app:proofs}

\begin{proof}[Proof of Theorem \ref{lemm:2sla}] By part (i) of Lemma \ref{lem:quad_exp},  $\sqrt{R}(\tilde\eta_R-\eta^\star)=O_p(1)$, under Assumption \ref{ass:samp}. Now, using Assumption \ref{ass:quad}(i), and Theorem 6 in \cite{andrews1999estimation}, expand the loss difference around $\theta^\star_R=(\eta^\star_R,\tilde\gamma_R')'$, $\eta^\star_R=\delta_R+\eta^\star$, where applicable, all derivatives are again taken to be left/right partial derivatives: for a sequence of intermediate values $\bar\eta_R$ satisfying $\|\bar\eta_R-\eta^\star_R\|\le\|\tilde\eta_R-\eta_R^\star\|$, we obtain
	\begin{flalign*}
		\sqrt{P}\Delta_P(\tilde\theta_R,\theta^\star_R)&=\{\nabla_\eta L_P(\eta^\star_R,\tilde\gamma_R)/{P}\}'\sqrt{P}(\tilde\eta_R-\eta^\star_R)+\frac{1}{2}\sqrt{P}(\tilde\eta_R-\eta^\star_R)'\{\nabla^2_{\eta\eta} L_P\left(\bar\eta_R,\tilde\gamma_R\right)/{P}\}(\tilde\eta_R-\eta^\star_R).
	\end{flalign*}From compactness of $\Theta$, Assumption \ref{ass:parm}, and differentiability of $\nabla_\eta L_R(\eta,\gamma^\star)$ in $\eta$, Assumption \ref{ass:quad}(i), the following are satisfied via the usual arguments: $$\|\nabla_\eta {L}_R(\eta^\star_R,\tilde\gamma_R)/{P}-\nabla_\eta \mathcal{L}(\eta^\star_R,\tilde\gamma_R)\|=o_p(1),\quad \|\nabla^2_{\eta\eta} L_P\left(\bar\eta_R,\tilde\gamma_R\right)/{P}-\mathcal{M}_{\eta\eta}(\bar\eta_R,\tilde\gamma_R) \|=o_p(1).$$From the above convergence, and the continuity of the left/right derivatives in Assumption \ref{ass:quad}, we can conclude that
	\begin{flalign}
		\sqrt{P}\Delta_P(\tilde\theta_R,\theta^\star_R)&=\nabla_\eta \mathcal{L}(\eta^\star_R,\gamma^\star)'\sqrt{P}(\tilde\eta_R-\eta^\star_R)+ (\tilde\eta_R-\eta^\star_R)'\mathcal{M}_{\eta\gamma}(\eta^\star_R,\gamma^\star)'\sqrt{P}(\tilde\gamma_R-\gamma^\star)+o_p(1)\nonumber
		\\&+\frac{1}{2}\sqrt{P}(\tilde\eta_R-\eta^\star_R)'\mathcal{M}_{\eta\eta}\left(\bar\eta_R,\gamma^\star\right)(\tilde\eta_R-\eta^\star_R)+o_p(1)\label{eq:lll}.
	\end{flalign}
	
	\noindent \textbf{Case (1)}. Consider that $\delta_R=\delta/T^\xi$, for $\xi\in[0,1/4)$,  then $$\sqrt{P}(\tilde\eta_R-\eta^\star_R)=\sqrt{P}(\tilde\eta_R-\eta^\star)+\delta_R\sqrt{P}=O_p(1)+\delta_R\sqrt{P},$$ where the $O_p(1)$ follows since,
	by Lemma \ref{lem:quad_exp}, $\|\tilde\eta_R-\eta^\star\|=O_p(1/\sqrt{P})$.  
	
	Consider two cases: $\xi=0$, and $\xi\in(0,1/4)$. In the case where $\xi=0$, from the definition of $\eta^\star_R$, for some $\varepsilon>0$, there exists an $T(\delta)$ large enough such that for all $T>T(\delta)$, $\|\eta^\star-\eta^\star_R\|\ge \gamma$ with $\gamma=\delta-\varepsilon>0$. Hence, by Assumption  \ref{ass:cons}(ii) and the differentiability in \ref{ass:quad}(i)
	$
	\plim_{T\rightarrow\infty}\|\nabla_\eta\mathcal{L}(\eta^\star_R,\gamma^\star)\|>0
	$ and $\plim_{T\rightarrow\infty}|\delta_R'\nabla_\eta\mathcal{L}(\eta^\star_R,\gamma^\star)|>0$. Alternatively, in the case where $\xi\in(0,1/4)$, there exists a $T(\delta)$ large enough such that for all $T>T(\delta)$, $\|\eta^\star-\eta^\star_R\|\le \gamma$, and we have  $\plim_{T\rightarrow\infty}|\delta_R'\nabla_\eta\mathcal{L}(\eta^\star_R,\gamma^\star)|=0$.
	
	Applying the above into equation \eqref{eq:lll} we have that
	\begin{flalign*}
		\sqrt{P}\Delta_P(\tilde\theta_R,\theta^\star_R)&=\sqrt{P}\delta_R'\{\nabla_\eta\mathcal{L}(\eta^\star_R,\gamma^\star)+o_p(1)+M^\star_{\eta\gamma}\left(\eta^\star_R,\gamma^\star\right)\sqrt{P}(\tilde\gamma_R-\gamma^\star)+\frac{1}{2}M^\star_{\eta\eta}\left(\bar\eta_R,\gamma^\star\right)\delta_R\}\\&=\sqrt{P}\delta_R'\{\nabla_\eta\mathcal{L}(\eta^\star_R,\gamma^\star)+O_p(1/\sqrt{P})\}+\frac{\sqrt{P}}{2}\delta_R'M^\star_{\eta\eta}\left(\bar\eta_R,\gamma^\star\right)\delta_R\\&\ge \sqrt{P}\delta_R'\{\nabla_\eta\mathcal{L}(\eta^\star_R,\gamma^\star)+O_p(1/\sqrt{P})\}+\frac{\sqrt{P}}{2}\inf_{\eta\in\mathcal{E}}\left\{\delta_R'M^\star_{\eta\eta}\left(\eta,\gamma^\star\right)\delta_R\right\},
	\end{flalign*}
	where the $O_p(1)$ term in the second equation follows by Assumption \ref{ass:bound}. Taking absolute values, and applying the reverse triangle inequality then yields
	\begin{flalign}\label{eq:new1}
		|\sqrt{P}\Delta_P(\tilde\theta_R,\theta^\star_R)|&\ge \sqrt{P}||\delta_R'\nabla_\eta \mathcal{L}(\eta^\star+\delta_R,\gamma^\star)|-|\inf_{\eta\in\mathcal{E}}\left\{\delta_R' M^\star_{\eta\eta}\left(\eta,\gamma^\star\right)\delta_R\right\}||+o_p(1)%\\&\ge \sqrt{P}\left||\delta'\nabla_\eta \mathcal{L}(\eta^\star+\delta^\star,\gamma^\star)|-\left|\delta'\inf_{\eta\in\mathcal{E}}[M^\star_{\eta\eta}\left(\eta,\gamma^\star\right)]\delta\right|\right|=
	\end{flalign}

	In the case where $\xi=0$, it can been seen directly that the term inside the absolute value is non-zero, and so the RHS diverges as $P\rightarrow\infty$. In the case where $\xi\in(0,1/4)$, rewrite 
	$$
	\sqrt{P}\delta_R=\sqrt{P}\delta/T^\xi=\sqrt{\frac{P}{P+R}}\frac{\sqrt{P+R}}{T^\xi}\delta=\sqrt{(1+c)^{-1}}T^{1/2-\xi}+o(1),
	$$where $c=\lim_T R/P$ is as in Assumption \ref{ass:samp}. Consider the second term on the RHS of \eqref{eq:new1}. Since $\xi<1/4$, the second term is proportional to $\sqrt{P}\|\delta_R\|^2=CT^{1/2-\xi}\rightarrow+\infty$ as $T\rightarrow+\infty$. Consequently, the RHS of \eqref{eq:new1} also diverges so long as $\xi\in(0,1/4)$.

	\medskip
	
	\noindent \textbf{Case (2)}. Applying the same steps as in case (1), we arrive at the inequalities
\begin{flalign*}
		\sqrt{P}\Delta_P(\tilde\theta_R,\theta^\star_R)\le &\sqrt{P}\delta_R'\{\nabla_\eta\mathcal{L}(\eta^\star_R,\gamma^\star)+O_p(1/\sqrt{P})\}+\sup_{\eta\in\mathcal{E}}\frac{\sqrt{P}}{2}\delta_R'M^\star_{\eta\eta}\left(\bar\eta_R,\gamma^\star\right)\delta_R\\&\ge\sqrt{P}\delta_R'\{\nabla_\eta\mathcal{L}(\eta^\star_R,\gamma^\star)+O_p(1/\sqrt{P})\}+\inf_{\eta\in\mathcal{E}}\frac{\sqrt{P}}{2}\delta_R'M^\star_{\eta\eta}\left(\bar\eta_R,\gamma^\star\right)\delta_R 
\end{flalign*}	As discussed in Case (1), when $\delta_R=\delta/T^\xi$, we have that $\sqrt{P}\delta_R=C\delta T^{1/2-\xi}$. However, in the case where $\xi\ge1/4$, we have that 
$$
\nabla_\eta\mathcal{L}(\eta^\star_R,\gamma^\star)=\mathcal{M}_{\eta\eta}(\eta^\star,\gamma^\star)\delta_R +o(\|\delta_R\|)
$$
Applying this into the above inequalities yields 
\begin{flalign}
	\sqrt{P}\Delta_P(\tilde\theta_R,\theta^\star_R)\le &\sqrt{P}\delta_R'\mathcal{M}_{\eta\eta}(\eta^\star,\gamma^\star)\delta_R+O_p(\|\delta_R\|)+o(\sqrt{P}\|\delta_R\|^2)+\sup_{\eta\in\mathcal{E}}\frac{\sqrt{P}}{2}\delta_R'M^\star_{\eta\eta}\left(\bar\eta_R,\gamma^\star\right)\delta_R\nonumber\\\ge&\sqrt{P}\delta_R'\mathcal{M}_{\eta\eta}(\eta^\star,\gamma^\star)\delta_R+O_p(\|\delta_R\|)+o(\sqrt{P}\|\delta_R\|^2)+\inf_{\eta\in\mathcal{E}}\frac{\sqrt{P}}{2}\delta_R'M^\star_{\eta\eta}\left(\bar\eta_R,\gamma^\star\right)\delta_R \label{eq:new2}
\end{flalign}Further, since $\sqrt{P}\|\delta_R\|^2=C>0$, 
\begin{flalign*}
	\sqrt{P}\Delta_P(\tilde\theta_R,\theta^\star_R)-C^2 \delta'\mathcal{M}_{\eta\eta}(\eta^\star,\gamma^\star)\delta\le &O_p(\|\delta_R\|)+o(\sqrt{P}\|\delta_R\|^2)+\sup_{\eta\in\mathcal{E}}\frac{\sqrt{C^2}}{2}\delta'\mathcal{M}_{\eta\eta}\left(\eta,\gamma^\star\right)\delta\\\ge&O_p(\|\delta_R\|)+o(\sqrt{P}\|\delta_R\|^2)+\inf_{\eta\in\mathcal{E}}\frac{\sqrt{C^2}}{2}\delta'\mathcal{M}_{\eta\eta}\left(\eta,\gamma^\star\right)\delta. 
\end{flalign*} Since $\mathcal{M}_{\eta\eta}(\eta,\gamma^\star)$ is continuous in $\eta$, the last terms in the inequalities are bounded. Letting $S_1=\sup_{\eta\in\mathcal{E}}\frac{\sqrt{C^2}}{2}\delta'\mathcal{M}_{\eta\eta}\left(\eta,\gamma^\star\right)\delta$ and $S_2=\inf_{\eta\in\mathcal{E}}\frac{\sqrt{C^2}}{2}\delta'\mathcal{M}_{\eta\eta}\left(\eta,\gamma^\star\right)\delta$ the above analysis has demonstrated that 
$$
S_2+o_p(1)\le \sqrt{P}\Delta_P(\tilde\theta_R,\theta^\star_R)-C^2 \delta'\mathcal{M}_{\eta\eta}(\eta^\star,\gamma^\star)\le S_1+o_p(1).
$$
Hence, depending on $S_1$ and $S_2$, we have show that $\lim_T\text{Pr}\left(|\sqrt{P}\Delta_P(\tilde\theta_R,\theta^\star_R)|>0\right)\lessgtr\alpha$. 
	\medskip
	
	\noindent \textbf{Case (3)}. With $\delta_R=\delta/T^\xi$, and $\xi>1/4$, we see that $\sqrt{P}\|\delta_R\|^2=C\|T^{1/2-2\xi}\|$. Hence, when $\xi>1/4$, $\sqrt{P}\|\delta_R\|^2=o(1)$. Applying this into the upper bound in equation \eqref{eq:new2} implies that 
\begin{flalign*}
		\sqrt{P}\Delta_P(\tilde\theta_R,\theta^\star_R)\le &\sqrt{P}\delta_R'\mathcal{M}_{\eta\eta}(\eta^\star,\gamma^\star)\delta_R+O_p(\|\delta_R\|)+o(\sqrt{P}\|\delta_R\|^2)+\sup_{\eta\in\mathcal{E}}\frac{\sqrt{P}}{2}\delta_R'M^\star_{\eta\eta}\left(\bar\eta_R,\gamma^\star\right)\delta_R\\&\le O_p(\sqrt{P}\|\delta_R\|^2)\{1+o_p(1)\}+O_p(\|\delta_R\|)\\&=o_p(1)	
\end{flalign*}

\end{proof}

\begin{proof}[Proof of Lemma \ref{lem:quads}]
	Recalling that $$\Delta_P(\vartheta^\delta_T,\tilde\theta_R)=P^{-1}\{L_P(\eta^\delta_T,\tilde\gamma_R)-L_P(\tilde\eta_R,\tilde\gamma_R)\}=-P^{-1}\left\{L_P(\tilde\eta_R,\tilde\gamma_R)-L_P(\eta^\delta_T,\tilde\gamma_R)\right\},$$ the result follows by applying Lemma \ref{lem:quad_exp}, in particular equation \eqref{eq:new_expans} obtained in the proof. For $\delta_T\asymp \delta/T^\xi$ with $\xi\in(0,\infty]$, equation \eqref{eq:new_expans}	implies that 
	\begin{flalign*}
		\Delta_P(\vartheta^\delta_T,\tilde\theta_R)&=\left(\sqrt{P}(\tilde\eta_R-\eta^\delta_T)'\mathcal{J}_{}^{1/2}\left[-\mathcal{J}^{-1/2}\left\{\nabla_\eta L_P(\theta^\star)/\sqrt{P}+\mathcal{M}_{\eta\gamma}\sqrt{P}(\tilde\gamma_R-\gamma^\star)\right\}\right]\right)/P\\&-\left\{\frac{1}{2}\sqrt{P}(\eta^\delta_T-\tilde\eta_R)'\mathcal{J}_{}\sqrt{P}(\eta^\delta_T-\tilde\eta_R)\right\}/P+R_P(\theta)/P,
	\end{flalign*}
	which, recalling the definition of $V_{P,R}$, can be re-arranged as
	\begin{flalign*}
		\Delta_P(\vartheta^\delta_T,\tilde\theta_R)&=(\tilde\eta_R-\eta^\delta_T)'\mathcal{J}_{}^{1/2}V_{P,R}-\frac{1}{2}(\eta^\delta_T-\tilde\eta_R)'\mathcal{J}_{}(\eta^\delta_T-\tilde\eta_R)+R_P(\theta)/P 
	\end{flalign*}
	Adding and subtracting $\frac{1}{2}\|V_{P,R}\|^2$, we can re-arrange the above equation as 
	\begin{flalign}
		\Delta_P(\vartheta^\delta_T,\tilde\theta_R)&=\frac{1}{2}\|V_{P,R}\|^2-\frac{1}{2}\|\mathcal{J}^{1/2}\left(\eta^\delta_T-\eta^\star\right)-\mathcal{J}^{1/2}\left(\tilde\eta_R-\eta^\star\right)-V_{P,R}\|^2+R_P(\theta)/P
		\label{eq:delta_exp}	
	\end{flalign}
	Under the maintained assumption that $R\asymp P$, from equation \eqref{eq:remainder} in the proof of Lemma \ref{lem:quad_exp}, 
	\begin{flalign*}
		R_P(\vartheta_P)/P&\le P^{-1}o_p\left\{1 +\|\sqrt{P}(\eta_T^\delta-\eta^\star)\|+\|\sqrt{P}\left(\eta^\delta_T-\eta^\star\right)\|^2+\|\sqrt{P}(\tilde\theta_R-\theta^\star)\|^2\right\}\\&\equiv o_p\left\{1+P^{-1/2}\|\eta^\delta_T-\eta^\star\|+\|\eta^\delta_T-\eta^\star\|^2+\|(\tilde\theta_R-\theta^\star)\|^2\right\},
	\end{flalign*}where the last term follows for any $\eta^\delta_T-\eta^\star=o_p(1)$. 
	
\end{proof}	

\begin{proof}[Proof of Corrolary \ref{cor:quads}]
	The result follows by appropriately manipulating the expansion in \eqref{eq:delta_exp}. Multiplying the RHS of equation \eqref{eq:delta_exp} by  $P$, and re-arranging terms yields 
	\begin{flalign*}
		\Delta_P(\vartheta^\delta_T,\tilde\theta_R)&=-\frac{1}{2}\|\mathcal{J}^{1/2}\sqrt{P}\left(\eta^\delta_T-\eta^\star\right)-\mathcal{J}^{1/2}\sqrt{P}\left(\tilde\eta_R-\eta^\star\right)-\sqrt{P}V_{P,R}\|^2+\frac{1}{2}\|\sqrt{P}V_{P,R}\|^2+R_P(\theta)
	\end{flalign*}Recalling the definition of $V_{P,R}$: (i) since $R\asymp P$, we have that $\sqrt{P} V_{P,R}=O_p(1)$, and a similar argument shows that $\|\sqrt{P}(\tilde\eta_R-\eta^\star)\|=\|\frac{\sqrt{P}}{\sqrt{R}}\sqrt{R}(\tilde\eta_R-\eta^\star)\|=O_p(1)$; (ii) since $\eta^\delta_T-\eta^\star=o_p(1)$ under the maintained assumptions, the remainder term vanishes as $T\rightarrow+\infty$. 
	
	Using the definitions of $X_P$ and $Z_P$ stated before the corollary, points (i) and (ii)  allow us to rewrite the display equation as 
	\begin{flalign*}
		P\cdot\Delta_P(\vartheta^\delta_T,\tilde\theta_R)=&-\frac{1}{2}\left\|\mathcal{J}^{1/2}\sqrt{P}(\eta^\delta_T-\eta^\star)-\left\{\mathcal{J}^{1/2}\sqrt{P}(\tilde\eta_R-\eta^\star)-\mathcal{J}^{-1/2}\left(\sqrt{P}X_P+c^{-1}\mathcal{M}_{\eta\gamma}\sqrt{R}Z_{R,\gamma}\right)\right\}\right\|^2\\&+\frac{1}{2}\|\mathcal{J}^{-1/2}\left(\sqrt{P}X_P+c^{-1}\mathcal{M}_{\eta\gamma}\sqrt{R}Z_{R,\gamma}\right)\|^2+o_p(1).
	\end{flalign*}
	
	\medskip 
	
	\noindent\textbf{Case (i).} When $\delta_T\asymp\delta/T^\xi$, with $\xi\in(0,1/2)$, we have that $\sqrt{P}(\eta^\delta_T-\eta^\star)\asymp \sqrt{P}T^{-\xi}\asymp T^{1/2-\xi}\rightarrow +\infty$ as $T\rightarrow+\infty$ under the maintained assumption on $P,R,T$. Consequently, since $$\left\{\mathcal{J}^{1/2}\sqrt{P}(\tilde\eta_R-\eta^\star)-\mathcal{J}^{-1/2}\left(\sqrt{P}X_P+\mathcal{M}_{\eta\gamma}\sqrt{P}Z_{R,\gamma}\right)\right\}=O_p(1),$$ we have that $P\cdot\Delta_P(\vartheta^\delta_T,\tilde\theta_R)\rightarrow+\infty$ as $T\rightarrow+\infty$.

	\medskip 
	
	\noindent\textbf{Case (ii).} Firstly, we note that when $\theta^\star\in\mathrm{Int}(\Theta)$, standard results on the asymptotic behavior of two-step estimators can be used to show that 
	$$
	\mathcal{J}^{1/2}(\tilde\eta_R-\eta^\star)=-\mathcal{J}^{-1/2}\{\nabla_\eta L_R(\theta^\star)/R+\mathcal{M}_{\eta\gamma}(\tilde\gamma_R-\gamma^\star)\}+o_p(1/\sqrt{R})
	$$
	Applying the definitions of $X_P$, and $Z_{R,\gamma}$, and the above equation implies that, recalling that $c=\lim_T P/R$, 
	\begin{flalign*}
		\sqrt{P}\mathcal{J}^{1/2}(\tilde\eta_R-\eta^\star)-\sqrt{P}V_{P,R}=&-\frac{1}{c}\mathcal{J}^{-1/2}\{X_R/\sqrt{R}+\mathcal{M}_{\eta\gamma}\sqrt{R}Z_{R,\gamma}\}+\mathcal{J}^{-1/2}\{X_P/\sqrt{P}+\frac{1}{c}\mathcal{M}_{\eta\gamma}\sqrt{R}Z_{R,\gamma}\}+o_p(1)\\=& c^{-1}\sqrt{R}V_{R,R}-\sqrt{P}V_{P,R}
	\end{flalign*}
	Hence, if $(c^{-1}\sqrt{R}V_{R,R},\sqrt{P}V_{P,R})\Rightarrow \mathcal{V}=(\mathcal{V}_1',\mathcal{V}_2')'$, which is guaranteed under Assumption \ref{ass:dist}, then, by the continuous mapping theorem,
	$$
	\mathcal{J}^{1/2}\sqrt{P}(\tilde\eta_R-\eta^\star)-\sqrt{P}V_{P,R}\Rightarrow \mathcal{V}_1-\mathcal{V}_2.
	$$
	
	When $\delta_T=\delta/T^{1/2}$, we have that $\sqrt{P}(\eta^\delta_T-\eta^\star)= \delta\sqrt{P}T^{-\xi}=\delta\left(\frac{P}{R+P}\right)^{1/2}$. From the maintained assumption on $R,P,T$, we have that $$\left(\frac{P}{R+P}\right)^{1/2}\rightarrow \{1/(1+c)\}^{1/2}\quad \text{ as }T\rightarrow+\infty.$$ Hence, applying the above two displayed equations yields 
	$$
	-\frac{1}{2}\|\mathcal{J}^{1/2}\sqrt{P}\left(\eta^\delta_T-\eta^\star\right)-\mathcal{J}^{1/2}\sqrt{P}\left(\tilde\eta_R-\eta^\star\right)-\sqrt{P}V_{P,R}\|^2\Rightarrow \| \{1/(1+c)\}^{1/2}\delta -\mathcal{V}_1+\mathcal{V}_2\|^2.
	$$
	Further, recalling the definition of $V_{P,R}$, and since $\sqrt{P}V_{P,R}\rightarrow \mathcal{V}_2$, we have that 
	$$
	+\frac{1}{2}\|\mathcal{J}^{-1/2}\left(\sqrt{P}X_P+c^{-1}\mathcal{M}_{\eta\gamma}\sqrt{R}Z_{R,\gamma}\right)\|^2\equiv \frac{1}{2}\|\sqrt{P}V_{P,R}\|^2\Rightarrow \|\mathcal{V}_2\|^2\equiv \frac{1}{2}\|X+c^{-1}\mathcal{M}_{\eta\gamma}Z_\gamma\|^2
	$$
	
	Since $(c^{-1}\sqrt{R}V_{R,R},\sqrt{P}V_{P,R})\Rightarrow \mathcal{V}=(\mathcal{V}_1',\mathcal{V}_2')'$, any continuous transformation of these components also converges in distribution. Hence, the stated result follows. 
	
	\medskip
	
	\noindent\textbf{Case (iii).} The result follows precisely as in case (ii) by taking $\delta=0$. 
\end{proof}

%In what follows, let us abuse notation and write $L_P(\theta):=\sum_{t=T}^{n-1}\ell_{t+1}(\theta)$.
\begin{proof}[Proof of Theorem \ref{thm:diff}]
By Assumption \ref{ass:samp}, $P\asymp R$ and we can apply the conclusion of Lemma \ref{lem:lem2} to obtain the following expansion for the loss differences of the two-step combinations:
\begin{flalign*}
\sqrt{P}\{ L_P(\tilde\theta_R)/P-L_P(\theta^\star)/P\}&=P^{-1/2}\{ L_P(\tilde\theta_R)-L_P(\theta^\star)\}\\&=\sqrt{P}(\tilde\eta_R-\eta^\star)'\nabla_\eta L_P(\eta^\star,\gamma^\star)/P+\sqrt{P}(\tilde\gamma_R-\gamma^\star)'\nabla_\gamma L_P(\eta^\star,\gamma^\star)/P+o_p(1+\|\sqrt{P}(\tilde\theta_R-\theta^\star\|^2)
\end{flalign*}However, By part (i) of Lemma \ref{lem:quad_exp}, we have $\sqrt{P}(\tilde\theta_R-\theta^\star)\asymp\sqrt{R}(\tilde\theta_R-\theta^\star)=O_p(1)$ (under Assumption \ref{ass:samp}). Applying the above, and the fact that $\{\nabla_\eta L_P(\theta^\star)\}/P=o_p(1)$, yields
\begin{flalign}
\sqrt{P}\{ L_P(\tilde\theta_R)/P-L_P(\theta^\star)/P\}&=o_p\{\|\sqrt{P}(\tilde\eta_R-\eta^\star)\|\}+\nabla_\gamma \mathcal{L}(\theta^\star)\sqrt{P}(\tilde\gamma_R-\gamma^\star)+o_p(\|\sqrt{P}(\tilde\gamma_R-\gamma^\star)\|)\nonumber\\&=\nabla_\gamma \mathcal{L}(\theta^\star)\sqrt{P}(\tilde\gamma_R-\gamma^\star)+o_p(1)\label{eq:1exp_o}.
\end{flalign}

From part (ii) of Lemma \ref{lem:lem1}, with $R\asymp P$,
\begin{flalign}\label{eq:2exp_o}
\sqrt{P}\{ L_P(\widehat\theta_R)/{P}-L_P(\theta^0)/{P}\}=O_p(1/\sqrt{P})
\end{flalign}
Subtracting the two expansion in \eqref{eq:1exp_o} and \eqref{eq:2exp_o}, we have that
	\begin{flalign*}
	\sqrt{P}\Delta_P(\tilde\theta_R,\widehat\theta_R)&=\sqrt{P}\left\{L_P(\tilde\theta_R)/{P}-L_P(\widehat\theta_R)/{P}\right\}\\&=\sqrt{P}\left\{[L_P(\tilde\theta_R)-{L}_P(\theta^\star)]/{P}-[L_P(\widehat\theta_R)-{L}_P(\theta^0)]/{P}\right\}+\sqrt{P}\{L_P(\theta^\star)/{P}-L_P(\theta^0)/{P}\}\\&=\sqrt{P}\left\{L_P(\theta^\star)/{P}-L_P(\theta^0)/{P}\right\}+O_p({P}^{-1/2})+o_p({P}^{-1/2})+\nabla_\gamma\mathcal{L}(\theta^\star)'\sqrt{P}(\tilde\gamma_R-\gamma^\star)\\&=\sqrt{P}\left\{\mathcal{L}(\theta^\star)-\mathcal{L}(\theta^0)\right\}+O_p({P}^{-1/2})+[1:\nabla_\gamma\mathcal{L}(\theta^*)']Z_P
	\end{flalign*} with $Z_P$ as defined in Assumption \ref{ass:dist}. By the hypothesis in Assumption \ref{ass:dist}, $[1:\nabla_\gamma\mathcal{L}(\theta^*)']{Z}_P$ is asymptotically normal with zero mean and variance $\Omega=[1:\nabla_\gamma\mathcal{L}(\theta^*)']Q[1:\nabla_\gamma\mathcal{L}(\theta^*)']'$.
	
	Now, define $\tilde Z_P:=[1:\nabla_\gamma\mathcal{L}(\theta^*)']Z_P/\sqrt{\Omega}$, and consider the probability
	\begin{flalign*}
	\mathrm{Pr}\left[\Delta_P(\tilde\theta_R,\widehat\theta_R)/\sqrt{\Omega}\le 0\right]&=\mathrm{Pr}\left[ \tilde{Z}_{P}+O_p({P}^{-1/2})+\sqrt{P}\left\{\mathcal{L}(\theta^\star)-\mathcal{L}(\theta^0)\right\}/\sqrt{\Omega}\le0\right]\\&=\mathrm{Pr}\left[ \tilde{Z}_{P}+O_p({P}^{-1/2})\le\frac{\sqrt{P}}{\sqrt{\Omega}}\left\{\mathcal{L}(\theta^0)-\mathcal{L}(\theta^\star)\right\}\right].
	\end{flalign*}Define $z_{P}:={\sqrt{P}}\left\{\left[\mathcal{L}(\theta^0)-\mathcal{L}(\theta^\star)\right]\right\}/\sqrt{\Omega}$ and note that, for any $P\ge 1$, $z_{P}<0$, since, by Assumption \ref{ass:cons}, $\mathcal{L}(\theta^0)< \mathcal{L}(\theta^\star)$, and $z_{P}\rightarrow-\infty$ as $P\rightarrow\infty$.

	Since $\tilde{Z}_{P}$ is asymptotically standard normal,
	\begin{flalign*}
	\mathrm{Pr}\left[\Delta_{P}(\tilde\theta_R,\widehat\theta_R)/\sqrt{\Omega}\le 0\right]&=\Phi(z_k)+\left\{\mathrm{Pr}(\tilde{Z}_{P}\le z_{P})-\Phi(z_k)\right\}+\left\{\mathrm{Pr}[\tilde{Z}_{P}\le z_{P}+o_p(1)]-\mathrm{Pr}(\tilde{Z}_{P}\le z_{P})\right\}\\&\leq \Phi(z_k)+\sup_{z}|\mathrm{Pr}(\tilde{Z}_{P}\le z_{P})-\Phi(z_k)|+o(1).
	\end{flalign*}Fix $\varepsilon>0$. Since $z_k\rightarrow-\infty$ as $P\rightarrow\infty$,  for some $P$ large enough we can conclude that $\Phi(z_k)\le\varepsilon/2$. From the convergence $\tilde{Z}_{P}\Rightarrow N(0,1)$, and the continuity of $\Phi(z)$, Polya's Theorem implies that for some $P$ large enough,
	$$
	\sup_{z}\left|\mathrm{Pr}(\tilde{Z}_{P}\le z)-\Phi(z)\right|\le \varepsilon/2,
	$$and for some $P$ large enough $\mathrm{Pr}[\Delta_{P}(\tilde\theta_R,\widehat\theta_R)/\sqrt{\Omega}\le 0]\le \varepsilon$. Since $\varepsilon>0$ is arbitrary, the results follows.
\end{proof}

\section{Proofs of Key Lemmas}

\begin{proof}[Proof of Lemma \ref{lem:lem1}]
	To prove part (i), we first prove consistency of $\widehat\theta_n$ for some $\vartheta\in\Theta_{\mathrm{I}}$, which can be proven by verifying the sufficient conditions in Theorem 2 of \cite{yuan1998asymptotics}. By \ref{ass:quad}(i), the function $\nabla_\theta L_n(\vartheta)=O_p(1/\sqrt{n})$ for each $\vartheta\in\Theta_{\mathrm{I}}$, which satisfies Assumption 1 of \cite{yuan1998asymptotics}. From Assumption \ref{ass:quad}(i), for any $\vartheta\in\Theta_{\mathrm{I}}$, there exists a neighbourhood $\mathcal{N}(\vartheta)$, such that for each $\theta\in\mathcal{N}(\vartheta)$, $L_n(\theta)$ has continuous second-order left/right partial derivatives, $\nabla_{\theta\theta}L_n(\theta)$, and by Assumption \ref{ass:quad}(iii), $\nabla_{\theta\theta}L_n(\theta)/n$ converges uniformly to $\nabla_{\theta\theta}\mathcal{L}(\theta)$, which is non-singular for each $\vartheta\in\Theta_{\mathrm{I}}$. Thus, Assumption 2 of \cite{yuan1998asymptotics} is satisfied, and we can conclude that $\widehat\theta_n:=\arg_{\theta\in\Theta}\{\nabla_\theta L_n(\theta)=0\}$ converges to $\vartheta$, for some $\vartheta\in\Theta_{\mathrm{I}}$. 
	
	Given that $\widehat\theta_n=\vartheta+o_p(1)$, for some $\vartheta\in\Theta_{\mathrm{I}}$, the remainder of the result follows similar arguments to Theorem 1 of \cite{andrews1999estimation}. In particular, from Theorem 6 in \cite{andrews1999estimation}  for l/r differentiable functions, the following Taylor series expansion is valid: for some $\vartheta\in\Theta_{\mathrm{I}}$, 
	\begin{flalign}
		L_n(\theta)-L_n(\vartheta)&=\sqrt{n}(\theta-\vartheta)'\nabla_{\theta}L_n(\vartheta)/\sqrt{n}+\frac{1}{2}\sqrt{n}(\theta-\vartheta)'\left[\nabla_{\theta\theta}L_n(\bar\theta)/n\right]^{}\sqrt{n}(\theta-\vartheta)'\nonumber\\&=\sqrt{n}(\theta-\vartheta)'\nabla_{\theta}L_n(\vartheta)/\sqrt{n}+\frac{1}{2}\sqrt{n}(\theta-\vartheta)'\left[\nabla_{\theta\theta}\mathcal{L}(\vartheta)\right]^{}\sqrt{n}(\theta-\vartheta)\nonumber\\&+\frac{1}{2}\sqrt{n}(\theta-\vartheta)'\left[\nabla_{\theta\theta}\mathcal{L}_n(\bar\theta)/n-\nabla_{\theta\theta}\mathcal{L}(\vartheta)\right]^{}\sqrt{n}(\theta-\vartheta)\label{eq:one-expand}
	\end{flalign}for $\bar\theta$ an intermediate value satisfying $\|\vartheta-\bar\theta\|\le\|\theta-\vartheta\|$. Clearly,
	\begin{flalign*}
		R_n(\theta):=\frac{1}{2}\sqrt{n}(\theta-\vartheta)'\left[\nabla_{\theta\theta}\mathcal{L}_n(\bar\theta)/n-\nabla_{\theta\theta}\mathcal{L}(\vartheta)\right]^{}\sqrt{n}(\theta-\vartheta)\le \frac{1}{2}\|\sqrt{n}(\theta-\vartheta)\|^2\|\nabla_{\theta\theta}\mathcal{L}_n(\bar\theta)/n-\nabla_{\theta\theta}\mathcal{L}(\vartheta)\|, 
	\end{flalign*}so that for any $\|\theta-\vartheta\|=o(1)$, we have that the remainder term in \eqref{eq:one-expand} is $O(\|\sqrt{n}(\theta-\vartheta)\|^2)$ by Assumption \ref{ass:quad}(iii). 
	
	Recalling the notations $\mathcal{M}_{\theta\theta}(\theta)=\nabla_{\theta\theta}\mathcal{L}(\theta)$, $\mathcal{J}=\mathcal{M}_{\theta\theta}(\vartheta)$, and ${Z}_n=-\mathcal{J}^{-1/2}\nabla_\theta L_n(\vartheta)/\sqrt{n}$, we obtain
	\begin{flalign*}
		L_n(\theta)-L_n(\vartheta)&=-\sqrt{n}(\theta-\vartheta)'\mathcal{J}_{}^{1/2}{Z}_n+\frac{1}{2}\sqrt{n}(\theta-\vartheta)'\mathcal{J}_{}\sqrt{n}(\theta-\vartheta)+R_n(\theta). %\\&=-\frac{1}{2}\left[\sqrt{n}(\eta-\eta^\star)-\widetilde{V}_n\right]'\mathcal{M}_{\eta\eta}(\theta^\star)^{}\left[\sqrt{n}(\eta-\eta^\star)-\widetilde{V}_n\right]+\frac{1}{2}\widetilde{V}_n'\mathcal{M}_{\eta\eta}(\theta^\star)\widetilde{V}_n+R_n(\theta)
	\end{flalign*}
	From the definition of $\widehat\theta_n$, we have $L_n(\vartheta)\ge L_n(\widehat\theta_n)$, so that applying the above equation yields
	$$
	0\ge L_n(\widehat\theta_n)-L_n(\vartheta)=-\sqrt{n}(\widehat\theta_n-\vartheta)'\mathcal{J}^{1/2}Z_n+\frac{1}{2}\sqrt{n}(\widehat\theta_n-\vartheta)'\mathcal{J}\sqrt{n}(\widehat\theta_n-\vartheta)+R_n(\widehat\theta_n).	
	$$The above equation is precisely the (negative of) equation (7.3) in the proof of Theorem 1 in \cite{andrews1999estimation} and the remainder of the proof follows the same argument. 

Having shown that $\sqrt{n}(\widehat\theta_n-\vartheta)=O_p(1)$, for some $\vartheta\in\Theta_{\mathrm{I}}$, the proof of part (ii) of the stated result follows the same argument as in Theorem 2 part (b) of \cite{andrews1999estimation}. 
	
\end{proof}

\begin{proof}[Proof of Lemma \ref{lem:quad_exp}]

To prove the first stated result, we use Theorem 6 in \cite{andrews1999estimation}  for l/r differentiable functions to obtain the following Taylor series expansion:
	\begin{flalign*}
	L_n(\eta,\tilde\gamma_n)-L_n(\eta^\star,\tilde\gamma_n)&=\sqrt{n}(\eta-\eta^\star)'\nabla_{\eta}L_n(\eta^\star,\tilde\gamma_n)/\sqrt{n}+\frac{1}{2}\sqrt{n}(\eta-\eta^\star)'\left[\nabla_{\eta\eta}L_n(\bar\eta,\tilde\gamma_n)/n\right]^{}\sqrt{n}(\eta-\eta^\star)'\\&=\sqrt{n}(\eta-\eta^\star)'\left\{\nabla_{\eta}L_n(\theta^\star)/\sqrt{n}+n^{-1}\nabla_{\eta\gamma}L_n(\eta^\star,\bar\gamma)[\sqrt{n}(\tilde\gamma_n-\eta^\star)]\right\}\\&+\frac{1}{2}\sqrt{n}(\eta-\eta^\star)'\left[n^{-1}\nabla_{\eta\eta}L_n(\bar\eta,\tilde\gamma_n)\right]^{}\sqrt{n}(\eta-\eta^\star)'
	\end{flalign*}for $\bar\eta$, $\bar\gamma$,  intermediate values satisfying $\|\eta^\star-\bar\eta\|\le\|\eta-\eta^\star\|$ and $\|\gamma^\star-\bar\gamma\|\le\|\tilde\gamma_n-\gamma^\star\|$, respectively. Letting
	\begin{flalign}
	R_n(\eta,\tilde\gamma_n):=&\sqrt{n}\left(\eta-\eta^\star\right)'\left\{\left[\nabla_{\eta\eta}L_n(\bar\eta,\tilde\gamma_n)/n\right]^{}-\left[\nabla_{\eta\eta}\mathcal{L}(\bar\eta,\tilde\gamma_n)\right]^{}\right\}\sqrt{n}\left(\eta-\eta^\star\right)\nonumber\\&-\sqrt{n}\left(\eta-\eta^\star\right)'\left\{\left[\nabla_{\eta\eta}\mathcal{L}(\theta^\star)\right]^{}-\left[\nabla_{\eta\eta}\mathcal{L}(\bar\eta,\tilde\gamma_n)\right]^{}\right\}\sqrt{n}\left(\eta-\eta^\star\right)\nonumber\\&+\sqrt{n}(\eta-\eta^\star)\{n^{-1}\nabla_{\eta\gamma}L_n(\eta^\star,\bar\gamma)-\nabla_{\eta\gamma}\mathcal{L}(\theta^\star)\}\sqrt{n}(\tilde\gamma_n-\gamma^\star), \label{eq:remain}
	\end{flalign}we have
	\begin{flalign}
	L_n(\eta,\tilde\gamma_n)-L_n(\eta^\star,\tilde\gamma_n)&=\sqrt{n}(\eta-\eta^\star)'\left\{\nabla_{\eta}L_n(\eta^\star,\tilde\gamma_n)/\sqrt{n}+\nabla_{\eta\gamma}\mathcal{L}(\theta^\star)\sqrt{n}(\tilde\gamma_n-\gamma^\star)\right\}\nonumber\\&+\frac{1}{2}\sqrt{n}(\eta-\eta^\star)'\left[\nabla_{\eta\eta}\mathcal{L}(\theta^\star)\right]^{}\sqrt{n}(\eta-\eta^\star)'+R_n(\eta,\tilde\gamma_n).\label{eq:2sexpand}
	\end{flalign}
	
	Recalling the notations $\mathcal{M}_{\eta\eta}(\theta)=\nabla_{\eta\eta}\mathcal{L}(\theta)$, $\mathcal{J}=\mathcal{M}_{\eta\eta}(\theta^\star)$, and ${V}_n=-\mathcal{J}^{-1/2}\{\nabla_\eta L_n(\theta^\star)/\sqrt{n}+\mathcal{M}_{\eta\gamma}(\theta^\star)\sqrt{n}(\tilde\gamma_n-\gamma^\star)\}$, we obtain
	\begin{flalign}
	L_n(\eta,\tilde\gamma_n)-L_n(\eta^\star,\tilde\gamma_n)&=-\sqrt{n}(\eta-\eta^\star)'\mathcal{J}_{}^{1/2}(\theta^\star){V}_n+\frac{1}{2}\sqrt{n}(\eta-\eta^\star)'\mathcal{J}_{}\sqrt{n}(\eta-\eta^\star)+R_n(\theta).\label{eq:new_expans} %\\&=-\frac{1}{2}\left[\sqrt{n}(\eta-\eta^\star)-\widetilde{V}_n\right]'\mathcal{M}_{\eta\eta}(\theta^\star)^{}\left[\sqrt{n}(\eta-\eta^\star)-\widetilde{V}_n\right]+\frac{1}{2}\widetilde{V}_n'\mathcal{M}_{\eta\eta}(\theta^\star)\widetilde{V}_n+R_n(\theta)
	\end{flalign}
	From the definition of $\tilde\eta_n$, we have $L_n(\eta^\star,\tilde\gamma_n)\ge L_n(\tilde\eta_n,\tilde\gamma_n)$, so that applying the above equation yields
	$$
	0\ge L_n(\tilde\eta_n,\tilde\gamma_n)-L_n(\eta^\star,\tilde\gamma_n)=-\sqrt{n}(\tilde\eta_n-\eta^\star)'\mathcal{J}^{1/2}V_n+\frac{1}{2}\sqrt{n}(\tilde\eta_n-\eta^\star)'\mathcal{J}\sqrt{n}(\tilde\eta_n-\eta^\star)+R_n(\tilde\theta_n).	
	$$

	Now, consider $R_n(\tilde\theta_n)=R_{1n}(\tilde\theta_n)+R_{2n}(\tilde\theta_n)+R_{3n}(\tilde\theta_n)$, corresponding to each of the three terms in \eqref{eq:remain}. From the consistency of $\tilde\theta_n$, there exist some $\delta_n=o(1)$ such that $\|\tilde\theta_n-\theta^\star\|\le\delta_n$ with probability converging to one. For the first term, we have
	$$
	|R_{1n}(\tilde\theta_n)|\le \|\sqrt{n}(\tilde\eta_n-\eta^\star)\|^2\sup_{\|\eta-\eta^\star\|\le\delta_n}\| \nabla_{\eta\eta}L_n(\bar\eta,\tilde\gamma_n)/n -\nabla_{\eta\eta}\mathcal{L}(\bar\eta,\tilde\gamma_n)\|=o_p\left(\|\sqrt{n}(\tilde\eta_n-\eta^\star)\|^2\right)
	$$where the equality follows from Assumption \ref{ass:quad}(iv). Similarly, for the third term we have
	$$
	|R_{3n}(\tilde\theta_n)|\le \|\sqrt{n}(\tilde\eta_n-\eta^\star)\|^2\sup_{\|\eta-\eta^\star\|\le\delta_n}\| \nabla_{\eta\gamma}L_n(\eta^\star, \bar\gamma)/n -\nabla_{\eta\gamma}\mathcal{L}(\eta^\star,\gamma^\star)\|=o_p\left(\|\sqrt{n}(\tilde\eta_n-\eta^\star)\|^2 \vee \|\sqrt{n}(\tilde\gamma_n-\gamma^\star)\|\right),
	$$where the $o_p(1)$ term follows by consistency of $\tilde\gamma_n$, the definition of the intermediate value, $\bar\gamma$, i.e., $\|\bar\gamma-\gamma^\star\|\le\|\gamma^\star-\tilde\gamma_n\|$, and continuity of the second derivatives (Assumption \ref{ass:quad}(i)).
	For the second term, $|R_{2n}(\tilde\theta_n)|$, a similar argument to the above shows
	$$
	|R_{2n}(\tilde\theta_n)|\le o_p(1)(1+\|\sqrt{n}(\tilde\eta_n-\eta^\star)\|\|\sqrt{n}(\tilde\gamma_n-\gamma^\star)\|)=o_p(1+\|\sqrt{n}(\tilde\eta_n-\eta^\star)\|)O_p(1)=o_p(\|\sqrt{n}(\tilde\eta_n-\eta^\star)\|)
	$$where the $O_p(1)$ term in the second equality follows by Assumption \ref{ass:bound}. Putting all the terms together yields
	\begin{equation}\label{eq:remainder}
	|R_{n}(\tilde\theta_n)|\le o_p\{1+\|\sqrt{n}(\tilde\eta_n-\eta^\star)\|+\|\sqrt{n}(\tilde\eta_n-\eta^\star)\|^2\}.	
	\end{equation}

	Now, note that, for $\vartheta_n=(\eta^\star,\tilde\gamma_n)^\top$, we have that
	$$
	n\cdot\Delta_n(\tilde\theta_n,\vartheta_n)\equiv L_n(\tilde\eta_n,\tilde\gamma_n)-L_n(\eta^\star,\tilde\gamma_n).
	$$
	Recalling $\mathcal{J}:=\mathcal{M}_{\eta\eta}(\theta^\star)$, let $\kappa_n:=\mathcal{J}^{1/2}\sqrt{n}(\tilde\eta_n-\eta^\star)$, where $\mathcal{J}^{1/2}$ exists by Assumption \ref{ass:quad}(iii). Using the definitions, $V_n,\kappa_n$ and equation \eqref{eq:remainder}, we can re-arrange $$0\ge L_n(\tilde\eta_n,\tilde\gamma_n)-L_n(\eta^\star,\tilde\gamma_n)\equiv 	n \cdot \Delta_n(\tilde\theta_n,\vartheta_n)$$ as
	\begin{flalign*}
	0\ge 	n \Delta_n(\tilde\theta_n,\vartheta_n)&=-\kappa_n'{V}_n+\frac{1}{2}\kappa_n'\kappa_n+o_p\{1+\|\sqrt{n}(\tilde\eta_n-\eta^\star)\|+\|\sqrt{n}(\tilde\eta_n-\eta^\star)\|^2\}\\&=-O_p(\|\kappa_n\|)+\|\kappa_n\|^2/2+o_p(1)\{1+\|\mathcal{J}^{-1/2}\kappa_n\|+\|\mathcal{J}_{}^{-1/2}\kappa_n\|^2\}\\&=-\{O_p(1)-o_p(1)\}\|\kappa_n\|+\|\kappa_n\|^2/2+o_p(\|\kappa_n\|^2)+o_p(1),\\&=-\zeta_n\|\kappa_n\|+\|\kappa_n\|^2/2+o_p(\|\kappa_n\|^2)+o_p(\|\kappa_n\|),
	\end{flalign*}where $\zeta_n$ denotes the $\{O_p(1)-o_p(1)\} $ term. Rearranging the LHS of the above equation we have
	\begin{flalign*}
	\zeta_n^2\ge [\|\kappa_n\|\{1+o_p(1)\}-\zeta_n]^2+o_p(1)
	\end{flalign*}and we obtain $\|\kappa_n\|\le O_p(1)$. Since $\mathcal{J}$ is non-singular, by Assumption \ref{ass:quad}(iv), there exists some $c>0$ such that $$O_p(1)\ge \|\kappa_n\|=\|\mathcal{J}^{1/2}\sqrt{n}(\tilde\eta_n-\eta^\star)\|\ge c\|\sqrt{n}(\tilde\eta_n-\eta^\star)\|$$ and part (i) of the result follows.

	To establish part (ii), we use part (i) and the expansion
	\begin{flalign*}
	L_n(\eta,\tilde\gamma_n)-L_n(\eta^\star,\tilde\gamma_n)&=-\sqrt{n}(\eta-\eta^\star)'\mathcal{J}_{}^{1/2}{V}_n+\frac{n}{2}(\eta-\eta^\star)'\mathcal{J}(\eta-\eta^\star)+R_n(\tilde\theta_n)+o_p(\|\eta-\eta^\star\|),%\\&=-\frac{1}{2}\left[\sqrt{n}(\eta-\eta^\star)-\widetilde{V}_n\right]'\mathcal{M}_{\eta\eta}(\theta^\star)^{}\left[\sqrt{n}(\eta-\eta^\star)-\widetilde{V}_n\right]+\frac{1}{2}\widetilde{V}_n'\mathcal{M}_{\eta\eta}(\theta^\star)\widetilde{V}_n+R_n(\theta)
	\end{flalign*}which, from part (i), we can re-arrange as
	\begin{flalign}
		n \cdot \Delta_n(\tilde\theta_n,\vartheta_n)=L_n(\tilde\eta_n,\tilde\gamma_n)-L_n(\eta^\star,\tilde\gamma_n)&=-\kappa_n' V_n+\frac{1}{2}\kappa_n'\kappa_n+o_p(1)\{1+\|\kappa_n\|+\|\kappa_n\|^2\}\nonumber\\&=\frac{1}{2}\|\kappa_n-{V}_n\|^2-\frac{1}{2}{V}_n'{V}_n+o_p(1).\label{eq:expands}
	\end{flalign}The $o_p(1)$ term follows since, by part (i) of the result and Assumption \ref{ass:bound}, we have $\|\tilde\theta_n-\theta^\star\|=O_p(1)$, so that we have $R_n(\tilde\theta_n)=o_p(1)$.

\end{proof}

\begin{proof}[Proof of Lemma \ref{lem:lem2}]
	Again, Theorem 6 of \cite{andrews1999estimation} allows us to expand $L_n(\eta,\gamma)$ around $\theta=\theta^\star$:
	\begin{flalign*}
	L_n(\eta,\gamma)-L_n(\eta^\star,\gamma^\star)&=\sqrt{n}(\eta-\eta^\star)'\nabla_\eta L_n(\theta^\star)/\sqrt{n}+\sqrt{n}(\gamma-\gamma^\star)'\nabla_\gamma L_n(\theta^\star)/\sqrt{n}\\&+\frac{n}{2}(\theta-\theta^\star)'\nabla_{\theta\theta}\mathcal{M}(\theta^\star)(\theta-\theta^\star)+R_n(\theta),%\\&=\sqrt{n}(\eta-\eta^\star)'\xi^\eta_n(\eta^\star,\gamma^\star)/\sqrt{n}+\sqrt{n}(\gamma-\gamma^\star)'\nabla_\gamma\mathcal{L}(\theta^\star)\\&+\sqrt{n}(\gamma-\gamma^\star)'\{\xi^\gamma_n(\eta^\star,\gamma^\star)/\sqrt{n}-\nabla_\gamma\mathcal{L}(\theta^\star)\}-\frac{n}{2}(\theta-\theta^\star)'\mathcal{M}(\theta^\star)(\theta-\theta^\star)+R_n(\theta).
	\end{flalign*}where the remainder term $R_n(\theta)$ is given by
	$$
	R_n(\theta)=\frac{n}{2}(\theta-\theta^\star)'\left[\nabla_{\theta\theta}L_n(\bar\theta)/n-\mathcal{L}(\bar\theta)\right](\theta-\theta^\star)+\frac{n}{2}(\theta-\theta^\star)'\left[\nabla_{\theta\theta}\mathcal{L}(\bar\theta)-\mathcal{L}(\theta^\star)\right](\theta-\theta^\star).
	$$and $\bar\theta$ an intermediate value that satisfies $\|\theta^\star-\bar\theta\|\le\|\theta-\theta^\star\|$.
	
	However, since $\|\sqrt{n}(\tilde\theta_n-\theta^\star)\|=O_p(1)$, by Lemma \ref{lem:quad_exp}(a), Assumption \ref{ass:quad}(iv) and the definition of the intermediate value, we have that
	$
	|R_n(\tilde\theta_n)|\le o_p\{1+\|\sqrt{n}(\tilde\theta_n-\theta^\star)\|^2\}.
	$ Consequently,
	\begin{flalign*}
	n^{-1/2}\{L_n(\tilde\theta_n)-L_n(\theta^\star)\}&=\sqrt{n}(\tilde\theta_n-\theta^\star)'\frac{1}{n}\begin{pmatrix}\nabla_\eta L_n(\theta^\star)\\\nabla_\gamma L_n(\theta^\star)\end{pmatrix}	+\frac{n^{-1/2}}{2}\|\mathcal{M}_{\theta\theta}(\theta^\star)^{1/2}\sqrt{n}(\tilde\theta_n-\theta^\star)\|^2+o_p(1/\sqrt{n})\\&=(\tilde\eta_n-\eta^\star)'\nabla_\eta L_n(\theta^\star)/\sqrt{n}+\sqrt{n}(\tilde\gamma_n-\gamma^\star)'\nabla_\gamma L_n(\theta^\star)/n+O_p(1/\sqrt{n})+o_p(1).
	\end{flalign*}From Assumption \ref{ass:quad}(ii), $\nabla_\eta L_n(\theta^\star)/\sqrt{n}=O_p(1)$, while from Lemma \ref{lem:quad_exp}, $\|\tilde\eta_n-\eta^\star\|=o_p(1)$, and the first term above is $o_p(1)$.

	By Assumption \ref{ass:quad}(iv), $\nabla_\gamma L_n(\theta^\star)/n$ converges in probability to $\nabla_\gamma \mathcal{L}(\theta^\star)$. Moreover, by \ref{ass:cons}(iii), $\theta^\star$ is only a solution to $\nabla_\eta \mathcal{L}(\theta)=0$, so that $\nabla_\gamma \mathcal{L}(\theta^\star)\ne0$. Therefore, we can conclude that
\begin{flalign*}
n^{-1/2}\{L_n(\tilde\theta_n)-L_n(\theta^\star)\}&=o_p(1)+\sqrt{n}(\tilde\gamma_n-\gamma^\star)'\nabla_\gamma \mathcal{L}(\theta^\star).
\end{flalign*}

\end{proof}

\section{Numerical Implementation Details}

\subsection{Extending the Findings of \protect\cite{Smith2009}}

\label{subsec:nidextsmith}

\begin{figure*}[ht]
\includegraphics[width = \textwidth]{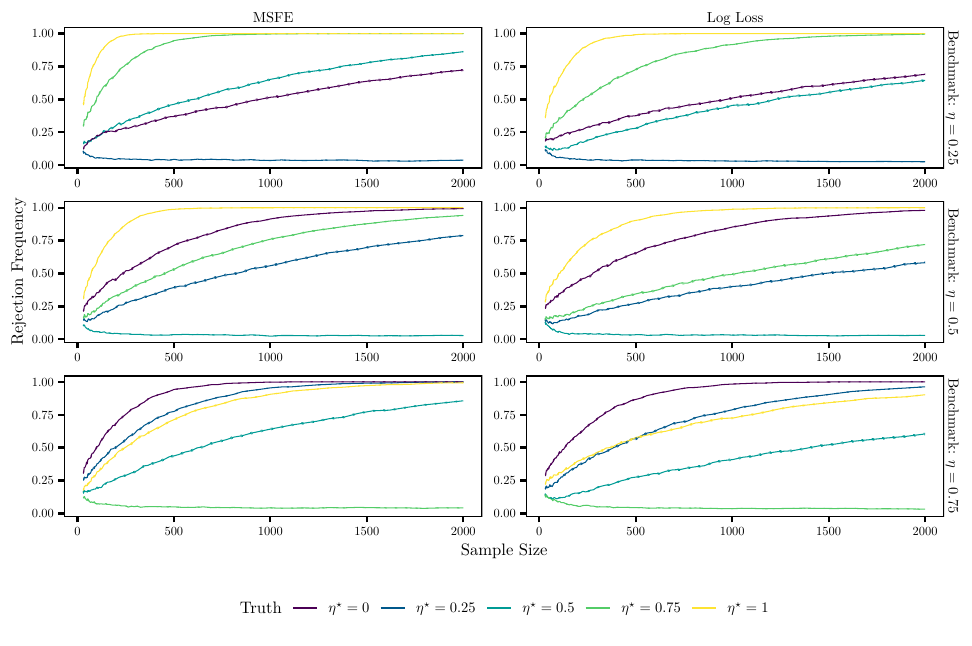}
\caption{Estimates (solid) and their 95\% confidence intervals (dotted) of
the Rejection Frequency ($y$-axis) for the hypothesis test of no inferior
predictive accuracy of a benchmark forecast combination with fixed weights
(rows) against the alternative optimal forecast combination. The test is
conducted with observations drawn from DGPs across a range of pseudo-true
weights (colors), and across a grid of sample sizes ($x$-axis). Results for
a point forecast combination with optimal weights minimizing the MSFE are
given in the first column, and results for a distributional forecast
combination with optimal weights minimizing the log loss are given in the
second column.}
\label{fig:smith2}
\end{figure*}

In this section, we detail the computational steps used to produce the
results in Section \ref{subsec:extsmith}, which extend the simulation
exercise given in Section 3.1 of \cite{Smith2009}. For convenience, the
results are reproduced in Figure \ref{fig:smith2}. Recall that this
simulation exercise is conducted based on data drawn from the following
zero-mean AR(2) process: 
\begin{equation*}
y_t = \phi_1 y_{t-1} + \phi_2 y_{t-2} + \epsilon_t,\ \epsilon_t \overset{%
i.i.d.}{\sim} N(0, \sigma^2_{\epsilon}),
\end{equation*}
where the data is used to estimate the parameters in the standard two-step
fashion, by minimizing the log loss or the MSFE of a distributional or point
forecast combination, respectively, based the following model: 
\begin{align*}
f^{(t)}_{1}(y) &= N\{y;\gamma_1 y_{t-1},1\}, \\
f^{(t)}_{2}(y) &= N\{y; \gamma_2 y_{t-2},1\}, \\
f^{(t)}(y) &= \eta f^{(t)}_{1}(y) + (1 - \eta) f^{(t)}_{2}(y),
\end{align*}
where $N\{x;\mu,\Sigma\}$ denotes the normal pdf evaluated at $x$ with mean $%
\mu$ and variance $\Sigma$, $\gamma_1$ and $\gamma_2$ are the parameters of
the constituent models, and $\eta$ is the weight assigned to the first model.

To produce the results, we must first produce the corresponding DGP
parameters given a selected value for the true optimal combination weight $%
\eta^{\star}$. Then for each such $\eta^{\star}$ value of interest, we draw
samples from the associated induced DGP to conduct the suite of tests
represented in Figure \ref{fig:smith2}, and iterate on these steps to
produce accurate estimates of the rejection frequency.

\subsubsection{Obtaining the DGP Parameters}

\label{subsubsec:dgpparams}

Given a desired value for $\eta^{\star}$, we conduct the following steps to
obtain DGP parameters $(\phi_1, \phi_2, \sigma^2_{\epsilon})$ for which the
actual value for $\eta^{\star}$ is as desired.

\begin{enumerate}
\item Draw $z_t \overset{i.i.d.}{\sim} N(0,1)$ for $t = 1, 2, \ldots, 10^7$.

\item Let $y_0 = y_{-1} = 0$ and $y_t = \phi_1 y_{t-1} + \phi_2 y_{t-2} +
\sigma_{\epsilon} z_t$ for $t = 1, 2, \ldots, 10^7$.

\item Solve the following equality-constrained optimization program: {\ 
{\footnotesize 
\begin{align}
\min_{\phi_1, \phi_2, \sigma_{\epsilon}^2, \gamma, \eta} \quad & - 1.1
\log(1 + \phi_2) - \log(1 - \phi_1 - \phi_2) - 2 \log(1 + \phi_1 - \phi_2) -
0.1 \log(\phi_1^2 + \phi_2^2)  \label{eqn:progcriterion} \\
\text{s.t.} \quad & \eta = \eta^{\star},  \label{eqn:progconst1} \\
& \mathrm{Var}(y_t) = 1,  \label{eqn:progconst2} \\
& \frac{1}{10^7} \sum_{t=1}^{10^7} \frac{\partial}{\partial \eta} S(\eta,
\gamma, y_t) = 0,  \label{eqn:progconst3} \\
& \frac{1}{10^7} \sum_{t=1}^{10^7} \frac{\partial}{\partial \gamma_j}
S_j(\gamma_j, y_t) = 0,\ j = 1, 2,  \label{eqn:progconst4}
\end{align}
}} where $S(\eta, \gamma, y_t)$ and $S_j(\gamma_j, y_t)$ is the score%
\footnote{%
A score, referring to either a scoring rule or scoring function, is a
measure of forecast accuracy. See Section \ref{sec:acc} for details.} of the
forecast combination and constituent forecast $j$, respectively, of $y_t$.
\end{enumerate}

We produce results for $S$ being either the MSFE, in which case 
\begin{equation*}
\begin{aligned} S(\eta, \gamma, y_t) &= (\mathbb{E}_{f^{(t)}}[y_t] - y_t)^2
\\ &= (\eta \mathbb{E}_{f^{(t)}_1}[y_t] + (1 - \eta)
\mathbb{E}_{f^{(t)}_2}[y_t] - y_t)^2 \\ &= (\eta \gamma_1 y_{t-1} + (1-\eta)
\gamma_2 y_{t-2} - y_t)^2, \\ S_j(\gamma_j, y_t) &=
(\mathbb{E}_{f^{(t)}_j}[y_t] - y_t)^2 \\ &= (\gamma_j y_{t-j} - y_t)^2,
\end{aligned}
\end{equation*}
or $S$ being the log loss, so that 
\begin{equation*}
{\footnotesize \begin{aligned} S(\eta, \gamma, y_t) &= - \log f^{(t)}(y_t)
\\ &= - \log ( \eta f^{(t)}_1(y_t) + (1 - \eta) f^{(t)}_2(y_t) ) \\ &= -
\log \bigg( \eta (2 \pi)^{-1} \exp \Big(-\frac{1}{2}(y_t - \gamma_1
y_{t-1})^2 \Big) + (1-\eta) (2 \pi)^{-1} \exp \Big(-\frac{1}{2}(y_t -
\gamma_2 y_{t-2})^2 \Big) \bigg), \\ S_j(\gamma_j, y_t) &= - \log
f^{(t)}_j(y_t) \\ &= - \log(2\pi) - \frac{1}{2} (y_t - \gamma_j y_{t-j})^2.
\end{aligned}}
\end{equation*}

The criterion function given in Equation \eqref{eqn:progcriterion} above is
designed to ensure that the DGP is strictly stationary and that $%
\eta^{\star} $ is identified. The DGP is stationary if and only if the roots
of its characteristic function lie inside the unit circle, which corresponds
to parameter values that satisfy $\phi_2 > -1$, $\phi_2 < 1 - \phi_1$ and $%
\phi_2 < 1 + \phi_1$. This space of parameter values is sometimes referred
to as the `stationarity triangle', because it forms a triangle in Euclidean
space. Inspecting the first three logarithmic terms in Equation %
\eqref{eqn:progcriterion} reveals that the criterion function is undefined
for parameter values that do not correspond to a stationary process, and
that the criterion function approaches infinity as the vector $(\phi_1,
\phi_2)$ approaches the boundary of the stationarity triangle. The final
logarithmic term in Equation \eqref{eqn:progcriterion} ensures that the
criterion function approaches infinity as the vector $(\phi_1, \phi_2)$
approaches zero, where $\eta^{\star}$ is not identified since $%
\gamma_1^{\star} = \gamma_2^{\star} = 0$. Figure \ref{fig:progcriterion}
displays the stationarity triangle and contours of the criterion function,
which is minimized at $(\phi_1, \phi_2) = (0.38, 0.14)$ when unconstrained.
Note that the constants $(-4, -1, -2, -0.1)$ with which the logarithmic
terms are multiplied in the criterion function are arbitrary, and any other
collection of negative constants would also suffice.

\begin{figure*}[ht]
\includegraphics[width = \textwidth]{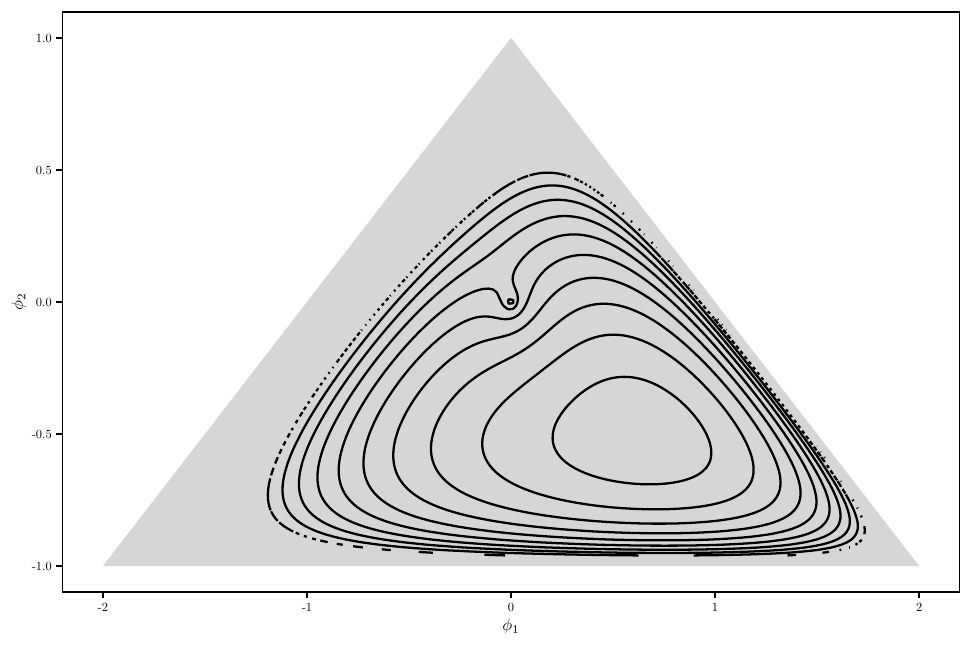}
\caption{The set of values for the AR(2) parameters $(\protect\phi_1, 
\protect\phi_2)$ for which the resulting stochastic process is strictly
stationary (gray), and contours of the criterion function in Equation 
\eqref{eqn:progcriterion} in a neighborhood of its minimum at $(\protect\phi%
_1, \protect\phi_2) = (0.38, 0.14)$ (black).}
\label{fig:progcriterion}
\end{figure*}

The constraints in Equations \eqref{eqn:progconst1} - \eqref{eqn:progconst4}
ensure that the desired value for the weight $\eta^{\star}$ is actually
obtained by the resulting DGP, and that the unconditional variance of the
DGP is identical across different values of $\eta^{\star}$. The value of the
weight itself is chosen in Equation \eqref{eqn:progconst1}, and the
unconditional variance of the DGP is fixed at unity by the constraint in
Equation \eqref{eqn:progconst2}. After solving the Yule-Walker equations for
the AR(2) process, we obtain the following expression for the unconditional
variance of $y_t$ as a function of the parameters $(\phi_1, \phi_2,
\sigma^2_{\epsilon})$: 
\begin{equation*}
\mathrm{Var}(y_t) = \frac{\sigma_{\epsilon}^2 (1 - \phi_2)}{(1 - \phi_1^2 -
\phi_2^2)(1 - \phi_2) - 2 \phi_1^2 \phi_2}.
\end{equation*}
The constraint in Equation \eqref{eqn:progconst3} is the first-order
condition for optimizing the weights according to the chosen score, across a
gigantic sample of $10^7$ observations. By the law of large numbers for
stationary processes, this average can approximate $\mathbb{E}%
_{DGP}[\partial S(\eta, \gamma, y_t)/\partial \eta]$ to any desired degree
of accuracy for a large enough sample size. Similarly, the constraint in
Equation \eqref{eqn:progconst4} ensures that the first-order condition for
optimizing the constituent model parameters are satisfied according to the
chosen score, across the same gigantic sample, and with a large enough
sample size the average in the constraint for the $j$\textsuperscript{th}
constituent model can approximate $\mathbb{E}_{DGP}[\partial S_j(\gamma_j,
y_t)/ \partial \gamma_j]$ to any desired degree of accuracy. Note that by
Assumptions \ref{ass:parm} - \ref{ass:quad}, $(\eta, \gamma) =
(\eta^{\star}, \gamma^{\star})$ is the solution to the problem of setting
the above expectations to zero, just as the constraints %
\eqref{eqn:progconst3} and \eqref{eqn:progconst4} set the corresponding
averages to zero.\footnote{%
See Section \ref{subsubsec:estforecomb} for more details on estimating
optimal forecast combinations, where Equations \eqref{eq:gamma} and %
\eqref{eq:est2s} define the two-step estimated combinations for which %
\eqref{eqn:progconst3} and \eqref{eqn:progconst4} are the first-order
conditions.} Since the sample size of $10^7$ is orders of magnitude larger
than the largest hypothesis test sample size of $2000$ considered in the
results displayed in Figure \ref{fig:smith2}, we assume that there is a
negligible difference between the desired weight imposed in Equation %
\eqref{eqn:progconst2} and the actual optimal weight $\eta^{\star}$
corresponding to the DGP which solves the optimization program.

We code the optimization program in R, and solve it using the \textit{nloptr}
package, a library of nonlinear optimization routines \citep{Johnson2022}.
In particular, we solve the optimization problem using sequential quadratic
programming, with an implementation based on \cite{Kraft1988, Kraft1994}.
See \cite{Gill2021} for a textbook treatment of sequential quadratic
programming.

\subsubsection{Estimating Rejection Frequency}

\label{subsubsec:estrejfreq}

Once we have a means of obtaining an AR(2) DGP with a selected pseudo-true
weight $\eta^{\star}$, producing the required estimates of the rejection
frequencies plotted in Figure \ref{fig:smith2} is straightforward. Briefly,
for each pseudo-true weight $\eta^{\star}$ (colors), each benchmark weight $%
\eta$ (rows), each score (columns), and each sample size ($x$-axis), we
conduct 5000 hypothesis tests of no inferior forecast accuracy of the
benchmark combination against the alternative optimal combination using 5000
samples of the required sample size drawn from a DGP with the required $%
\eta^{\star}$. The rejection frequency ($y$-axis) is then estimated as the
proportion of those tests that reject the null hypothesis. In detail, the
steps are as follows.

\begin{enumerate}
\item Obtain the DGP parameters $(\phi_1, \phi_2, \sigma^2_{\epsilon})$
corresponding to each desired pseudo-true weight $\eta^{\star} \in \{0,
0.25, 0.5, 0.75, 1\}$ for each score $S \in \{\mathrm{MSFE}, \mathrm{log\
loss}\}$ using the method detailed in Appendix \ref{subsubsec:dgpparams}
above.

\item For each score $S = \{\mathrm{MSFE}, \mathrm{log\ loss}\}$, each
pseudo-true weight $\eta^{\star} \in \{0, 0.25, 0.5, 0.75, 1\}$ and each
benchmark weight $\eta \in \{0.25, 0.5, 0.75\}$:

\begin{enumerate}
\item Draw the sample $y^{(1)}_{1:2000}$ from the DGP corresponding to $S$
and $\eta^{\star}$ that was found in Step 1.

\item For each such sample, consider the truncated samples $y^{(1)}_{1:T+1}$
for $T+1 = 30, 32, \ldots, 2000$, and for each of these conduct a hypothesis
test of no inferior forecast accuracy of: a) the benchmark forecast
combination with the weight fixed at $\eta$, against b) the alternative
forecast combination with optimal weights, where c) forecast accuracy is
measured according to $S$, and d) the in-sample and out-of-sample sizes are
equal. Record $R^{(1)}_{T+1} = 1$ where the test is rejected at the $5\%$
level, and $R^{(1)}_{T+1} = 0$ otherwise. This hypothesis test is described
in Section \ref{subsubsec:tests}, where $R = P = (T+1)/2$.
\end{enumerate}

\item Repeat Step 2 5000 times, drawing 5000 samples $y^{(1)}_{1:T+1},
y^{(2)}_{1:T+1}, \ldots, y^{(5000)}_{1:T+1}$ and obtaining hypothesis test
results $R^{(i)}_{T+1}$ for each $i = 1, 2, \ldots, 5000$, $T+1 = 30, 32,
\ldots, 2000$, $S$, $\eta^{\star}$ and $\eta$.

\item Estimate the rejection frequencies with the statistic $\hat{r}_{T+1} = 
\frac{1}{5000} \sum_{i=1}^{5000} R^{(i)}_{T+1}$, and calculate the
corresponding $95\%$ confidence intervals using standard asymptotics for
i.i.d. Bernoulli random variables.
\end{enumerate}

\subsection{Cause of the puzzle}

\label{subsec:nidpuzzlecause}

\begin{figure*}[ht]
\includegraphics[width = \textwidth]{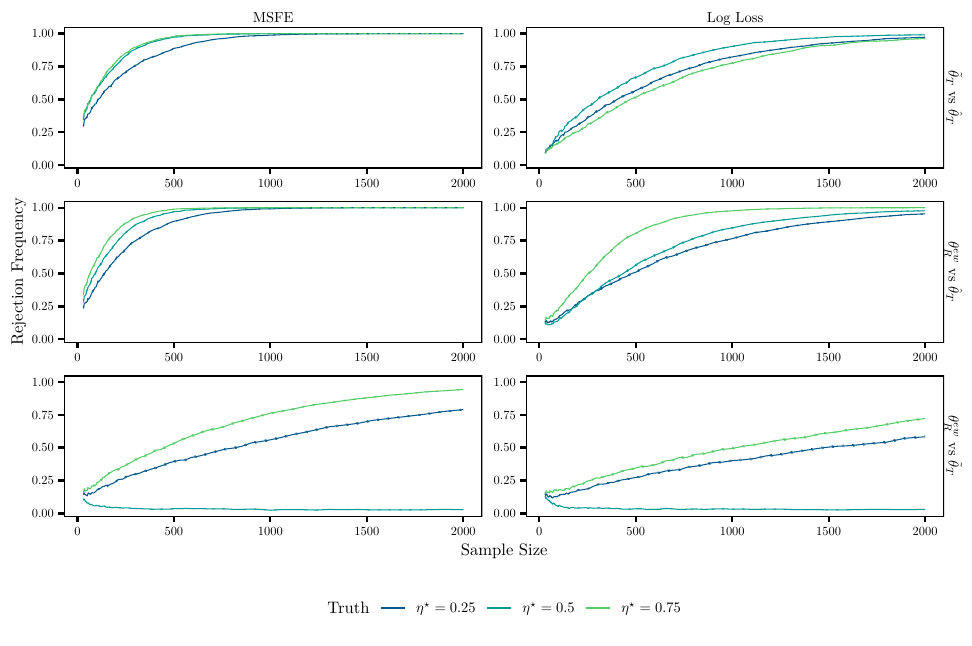}
\caption{Estimates (solid) and their 95\% confidence intervals (dotted) of
the Rejection Frequency ($y$-axis) for the hypothesis test of no inferior
predictive accuracy of a benchmark forecast combination against an
alternative combination (rows, benchmark vs alternative). The test is
conducted with observations drawn from DGPs across a range of pseudo-true
weights (colors), and across a grid of sample sizes ($x$-axis). Results for
a point forecast combination with optimal weights minimizing the MSFE are
given in the first column, and results for a distributional forecast
combination with optimal weights minimizing the log loss are given in the
second column.}
\label{fig:avoid2}
\end{figure*}

In this section, we address the computational steps used to produce the
results in Section \ref{subsec:smithone}, which are re-displayed in Figure %
\ref{fig:avoid2} for convenience. These steps are almost identical to those
described in Section \ref{subsec:nidextsmith} above, with one difference.
Whereas in Section \ref{subsubsec:estrejfreq} we conduct the hypothesis test
using a benchmark two-step combination with fixed weights $\eta \in \{0.25,
0.5, 0.75\}$ and an optimally-weighted two-step alternative, here we use the
(benchmark, alternative) pairs $\{(\tilde{\theta}_T, \hat{\theta}_T),
(\theta^{ew}_R, \hat{\theta}_T), (\theta^{ew}_R, \tilde{\theta}_T)\}$, which
you can see in the labels for the rows of Figure \ref{fig:avoid2}.

\end{document}